\documentclass[a4paper,12pt]{report}
\usepackage{graphicx}

\usepackage{longtable}
\usepackage{epsfig}
\usepackage[utf8]{inputenc}
\usepackage[T1]{fontenc}
\usepackage[french]{babel}
\usepackage{fancyhdr}
\usepackage[ruled,vlined,french,titlenumbered]{algorithm2e}
\usepackage{amsmath}
\usepackage{amssymb}
\usepackage{breakcites}
\usepackage{url}
\usepackage{multirow}

\newtheorem{conjecture}{Conjecture}

\newtheorem{definition}{Définition}
\newtheorem{example}{Exemple}

\newtheorem{lemma}{Lemma}

\newtheorem{proposition}{Proposition}
\newtheorem{property}{Propriété}
\newtheorem{proof}{Preuve}

\begin{document}

\begin{titlepage}
	\center\footnotesize RÉPUBLIQUE TUNISIENNE\\
	\vspace{0.1em}\footnotesize MINISTÈRE DE L'ENSEIGNEMENT SUPÉRIEUR ET DE LA RECHERCHE SCIENTIFIQUE\\
	\vspace{0.1em}\footnotesize UNIVERSITÉ DE TUNIS EL MANAR\\
	\vspace{0.1em}\footnotesize FACULTÉ DES SCIENCES DE TUNIS\vspace{0.1em}

	{\bfseries\huge Thèse} \\
	\vspace{0.25em}
	{\slshape\normalsize Présentée en vue de l’obtention du diplôme de } \\
	\vspace{0.25em}
	{\textrm{\normalsize\textbf{Docteur en Informatique}}} \\ 
	\vspace{1em}
	{\slshape\normalsize Par } \\
	\vspace{1em}
	{\bfseries\large Mohamed Nidhal \textsc{Jelassi}} \\
	\vspace{3em}\hrule
	
	{\textrm{\Huge\textbf{Etude, représentation et applications des traverses minimales d'un hypergraphe}}}
	
	\hrule
	
	\vspace{3em}
	\textbf{Soutenue publiquement le 08/12/2014 } \\ \vspace{1em}
	
	\textbf{Au sein du laboratoire : LIPAH} \\ \vspace{1em}
		\textbf{Sous la direction de Pr. Sadok Ben Yahia et Pr. Christine Largeron}
		
		\vspace{3.92em}
		
		\center\large{2013-2014}
	\end{titlepage}

\markboth{Remerciements}{Remerciements}
\fancyhead[R]{}
\footnotesize{
Je tiens tout d'abord à remercier Monsieur Jean-Marc \textsc{Petit}, Professeur à l'Institut National des Sciences Appliquées de Lyon, et Madame Sawssen \textsc{Krichen}, Maître de Conférences HDR à l'Institut Supérieur de Gestion, d'avoir accepté d'être les rapporteurs de mon travail de thèse. Leurs lectures minutieuses et leurs remarques pertinentes m'ont permis d'améliorer la qualité de ce manuscrit. Plus généralement, je remercie l'ensemble du jury, notamment Messieurs Faouzi \textsc{Ben Charrada}, Maître de Conférences HDR à la Faculté des Sciences de Tunis, François \textsc{Jacquenet}, Professeur à l'Université Jean Monnet de Saint-Etienne, et Mohamed \textsc{Quafafou}, Professeur à l'Université  d'Aix-Marseille, qui ont tout de suite accepté d'être examinateurs et juger mon travail.\\

Je  tiens  à  adresser  mes  vifs  remerciements  et  ma  profonde  reconnaissance  à  Madame Christine \textsc{Largeron}, Professeur a l'Université Jean Monnet de Saint-Etienne, et Monsieur Sadok \textsc{Ben Yahia}, Professeur à la Faculté des Sciences de Tunis, pour avoir dirigé mon travail, pour leurs conseils et disponibilités, pour leurs encouragements et leurs soutiens continus. L'intérêt qu'ils ont manifesté pour mon travail, leurs suggestions, et leurs remarques ont été d'une importance capitale. Sous leur direction, j'ai appris davantage de rigueur, de sens critique et de discipline.\\

Un merci spécial également à mon frère Nader \textsc{Jelassi}. Ta présence à mes côtés a été précieuse et vitale, notamment lors de nos séjours en France. Nos longues discussions, en arpentant les rues de Clermont et Saint-Etienne, ont été importantes pour l'avancement de cette thèse. Merci d'avoir toujours été là pour moi.\\

Un immense merci aussi pour mon compagnon de route, mon ami et mon frère Aymen \textsc{Sellaouti}. Nous nous sommes rencontrés lors de notre première année universitaire et nous ne nous sommes plus quittés, gravissant les échelons ensemble. Je te remercie pour tes encouragements, tes conseils et le soutien que tu m'as apporté tout au long de ces années.\\

Enfin, je remercie tous les membres du laboratoire \textsc{Lipah} pour ces années passées ensemble ainsi que ceux du laboratoire Hubert Curien pour l'accueil chaleureux qui m'a été réservé lors de mes différents séjours à Saint-Etienne.
}
\normalsize{
}

\newpage
 \pagenumbering{roman}
\thispagestyle{empty} \tableofcontents
\newpage
\thispagestyle{empty} \listoffigures
\newpage
\thispagestyle{empty} \listoftables
\newpage

\pagenumbering{arabic}
\chapter*{Introduction générale} \markboth{Introduction
générale}{\textsc{Introduction générale}}

La théorie des hypergraphes se propose de généraliser la théorie des graphes en introduisant le concept d'hyperarête. Une traverse est définie comme un ensemble de sommets qui intersecte toutes les hyperarêtes d'un hypergraphe et elle est dite minimale si elle l'est au sens de l'inclusion. Le problème de l'extraction des traverses minimales d'un hypergraphe est connu comme étant particulièrement difficile dans la mesure où premièrement, il est connu pour être coNP-complet malgré que sa complexité exacte est une question qui reste toujours ouverte, et deuxièmement de tous les algorithmes qui se sont attachés à calculer les traverses minimales, aucun n'a, à ce jour, une complexité théorique polynomiale en la taille de l'entrée et de la sortie, sauf pour des hypergraphes bien particuliers. De plus, l'un des verrous scientifiques les plus difficiles à affronter, en matière d'extraction des traverses minimales, est le nombre de traverses minimales à explorer qui peut être très élevé même pour des hypergraphes simples.

Pour autant, l'intérêt pour l'extraction des traverses minimales est en nette croissance due principalement aux solutions qu'elles offrent dans divers domaines d'application comme les bases de données, l'intelligence artificielle, l'e-commerce, le web sémantique, etc. Par ailleurs, la fouille de données arrive maintenant à maturité sur les contextes d'extraction classiques pour lesquels les algorithmes ont été mis au point. Ainsi, par exemple, les bases de données commerciales, qui décrivent les achats réalisés par des millions de clients sur des milliers de références sont désormais parfaitement exploitées par les spécialistes à l'aide de méthodes fondées sur les règles d'association. Ces techniques se popularisent aussi dans les domaines médicaux, économiques ou industriels.

Dans ce travail de thèse, nous focalisons notre intérêt sur les liens, déjà prouvés dans la littérature, entre la fouille de données et la théorie des hypergraphes pour redéfinir les notions d'hypergraphe et de traverse minimale. L'adaptation des techniques de la fouille de données, conjuguée à l'exploitation de certaines propriétés des hypergraphes, présente une voie intéressante pour la mise en place d'un cadre méthodologique pour l'optimisation de l'extraction des traverses minimales.

Avec l'emergence du Web 2.0 et des réseaux sociaux, nous avons tout d'abord été amenés à mettre à profit les traverses minimales pour la recherche d'information au sein ces systèmes communautaires. Ceci nous a conduit à proposer une modélisation originale d'un réseau social sous forme d'hypergraphe, particulièrement utile et adapté au cas où on ne dispose pas de toutes les relations entre les individus du réseau social considéré. A partir de cet hypergraphe, nous nous sommes intéressés à une classe particulière de traverses minimales, appelée \emph{traverse minimale multi-membres}, dans le contexte des systèmes communautaires. Un protocole expérimental basé sur des jeux de données du monde réel a confirmé l'intérêt de cette approche. Ce faisant, nous avons été confrontés au problème du nombre important de traverses minimales à extraire même pour un hypergraphe simple. Pour le résoudre, nous préconisons de représenter cet ensemble de manière concise et exacte en exploitant l'irrédondance de l'information dans les hypergraphes. De ce fait, notre travail consiste en la representation de l'ensemble des traverses minimales par un sous-ensemble succinct, composé de traverses minimales irrédondantes. L'introduction d'une mesure d'évaluation, appelée \emph{taux de compacité}, nous permettra de calculer le pourcentage de traverses minimales pouvant être déduite directement à partir de l'ensemble des traverses minimales irrédondantes. Nous avons illustré l'intérêt de cette représentation concise et exacte des traverses minimales pour résoudre le problème de l'inférence des dépendances fonctionnelles dans le but de calculer la couverture minimale d'une relation donnée.

Par ailleurs, afin d'optimiser le calcul des traverses minimales d'un hypergraphe, nous avons proposé de décomposer l'hypergraphe d'entrée en un nombre d'hypergraphes partiels égal au nombre de transversalité de l'hypergraphe initial. A partir de ces hypergraphes partiels, nous calculons leurs \emph{traverses minimales locales} respectives, dont le produit cartésien nous fournira un ensemble de traverses de l'hypergraphe. Les tests de minimalités permettent, ensuite, de ne garder que les traverses minimales. Le principal intérêt de cette approche est que ces tests sont inutiles pour les traverses dont la taille est égal au nombre de transversalité de l'hypergraphe d'entrée et dont nous sommes sûrs qu'ils sont minimales.

\bigskip

Le présent mémoire, décrivant le travail réalisé au cours de cette thèse, est composé de cinq chapitres.\\

Le \textbf{premier chapitre} introduit le contexte de nos recherches et présente les concepts de base que nous utiliserons dans la suite. La diversité des domaines d'application des traverses minimales est aussi mise en avant dans ce chapitre avec des points d'orgue pour les domaines ayant fait l'objet de nombreux travaux dans la littérature.\\

Le \textbf{deuxième chapitre} présente l'état de l'art et notamment les différents algorithmes, proposés dans la littérature, pour l'extraction des traverses minimales d'un hypergraphe. Ce chapitre mettra en évidence les différentes approches, techniques et autres stratégies utilisés pour présenter des solutions à cette problématique ainsi qu'une étude comparative de ces algorithmes. Cette synthèse permet aussi de situer nos contributions par rapport aux travaux antérieurs.\\

Notre première contribution est introduite dans le \textbf{troisième chapitre}. Nous présentons les \emph{traverses minimales multi-membres} (\textsc{tmm}), qui représentent une "sous-classe" des traverses minimales d'un hypergraphe. Ces \textsc{tmm} sont les plus petites traverses minimales, vérifiant une propriété de recouvrement, d'un hypergraphe d'entrée dont chaque hyperarête représente une communauté d'un réseau social donné. Un algorithme performant d'extraction de ces \textsc{tmm}, décrit dans ce chapitre, repose sur le calcul du nombre de transversalité de l'hypergraphe d'entrée. De plus, une application sur des jeux de données du monde réel est décrite et interprétée pour mettre en exergue l'intérêt des \textsc{tmm}.\\

Dans le \textbf{quatrième chapitre}, nous présentons notre deuxième contribution qui consiste en la representation concise et exacte de l'ensemble des traverses minimales. Cette représentation exploite l'irrédondance de l'information dans les hypergraphes, qui nous a conduit à définir et à mettre en place un cadre méthodologique pour le calcul des traverses minimales irrédondantes. Nous montrons, ensuite, l'intérêt de notre representation concise et exacte des traverses minimales à travers la résolution du problème de l'inférence des dépendances fonctionnelles. Les traverses minimales ayant déjà été utilisées pour optimiser le processus de calcul de la couverture minimale de toutes les dépendances fonctionnelles d'une relation donnée, nous proposons d'utiliser les traverses minimales irrédondantes pour réduire la taille de cette couverture.\\

Dans le \textbf{cinquième chapitre}, nous nous intéressons à l'extraction de toutes les traverses minimales. Étant donné que le nombre de ces dernières peu être exponentiel en la taille de l'hypergraphe, nous proposons d'optimiser leur calcul en décomposant l'hypergraphe d'entrée en des hypergraphes partiels. Nous choisissons un nombre d'hypergraphes partiels égal au nombre de transversalité de l'hypergraphe d'entrée afin d'éliminer des tests de minimalité et d'optimiser les temps de traitement nécessaires au calcul de toutes les traverses minimales. A partir de chaque hypergraphe partiel, un ensemble de traverses minimales (dites \emph{locales}) est calculé et un produit cartésien de toutes les traverses minimales locales permet de retrouver les traverses minimales de l'hypergraphe initial, suivant une stratégie "\emph{diviser pour régner}". \\ 
\chapter{Contexte et concepts de base}\label{chap1}
\section{Introduction}
Les hypergraphes généralisent la notion de graphe en définissant des hyperarêtes qui contiennent des familles de sommets, contrairement aux arêtes classiques qui ne joignent que deux sommets. D'un point de vue théorique, les hypergraphes permettent de généraliser certains théorèmes de graphes, voire d'en factoriser plusieurs en un seul. D'un point de vue pratique, ils sont parfois préférés aux graphes puisqu'ils modélisent mieux certains types de contraintes. Dans ce chapitre, nous présentons quelques définitions essentielles sur les hypergraphes et les traverses minimales nécessaires à l'introduction de la problématique de l'extraction des traverses minimales, en se basant essentiellement sur les définitions proposées par Berge dans \cite{berge1989}. Ensuite, nous passons en revue le large éventail des domaines d'application des traverses minimales.
\section{Préliminaires}
La théorie des hypergraphes se propose de généraliser la théorie des graphes en introduisant le concept d'hyperarête. Dans un hypergraphe où chaque hyperarête peut contenir plusieurs sommets, une traverse minimale correspond à un sous-ensemble de sommets qui intersecte toutes les hyperarêtes d'un hypergraphe, en étant minimal au sens de l'inclusion.

\subsection{Hypergraphes}

Un hypergraphe $H = (\mathcal{X}, \xi)$ est donc constitué de deux ensembles $\mathcal{X}$ et $\xi$, et est défini suivant la Définition \ref{hypergraphe}.

\begin{definition} \label{hypergraphe}
\textsc{hypergraphe} \cite{berge1989}\\
Soit le couple $H = (\mathcal{X}, \xi)$ avec $\mathcal{X}$ = $\{$$x_1$, $x_2$, \ldots, $x_n$$\}$ un ensemble fini et
$\xi$ = \{$e_1$, $e_2$, \ldots, $e_m$ \} une famille de parties de $\mathcal{X}$. $H$ constitue un hypergraphe sur $\mathcal{X}$ si :
\begin{enumerate}
    \item $e_i \neq \emptyset, i \in \{1$,\ldots, $m$\};
    \item $\bigcup \limits_{i=1, \ldots, m}~ e_i = \mathcal{X}$.
\end{enumerate}
\end{definition}

Les éléments $x_i$ de $\mathcal{X}$ sont appelés sommets de l'hypergraphe et les éléments $e_i$ de $\xi$ sont appelés hyperarêtes de l'hypergraphe.

Un hypergraphe est dit d'ordre \emph{n} si $\mid \mathcal{X} \mid$ = $n$ et la taille d'un hypergraphe est égale au nombre d'occurrences des sommets dans ses hyperarêtes.

\begin{example}
\bigskip
La Figure \ref{hypchap1} illustre un hypergraphe  $H = (\mathcal{X}, \xi)$ d'ordre $8$ et de taille $15$ tel que $\mathcal{X} = \{1, 2, 3, 4, 5, 6, 7, 8\}$ et  $\xi = \{\{$1, 2$\}$, $\{$2, 3, 7$\}$, $\{$3, 4, 5$\}$, $\{$4, 6$\}$, $\{$6, 7, 8$\}$, $\{$7$\}$\}.
\end{example}

\begin{figure}[htbp]

\begin{center}
\includegraphics[scale=0.8]{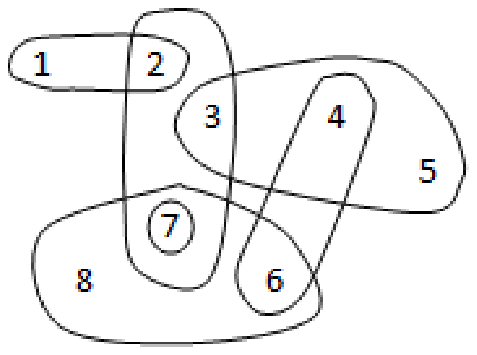}
\caption{Exemple d'un hypergraphe}\label{hypchap1}
\end{center}

\end{figure}

\begin{definition}\label{hsimple} \textsc{hypergraphe simple}\\
$H = (\mathcal{X}, \xi)$ est dit hypergraphe \emph{simple} si pour tout $e_i$ $\in$ $\xi$ et $e_j$ $\in$ $\xi$, alors $e_i$ $\subseteq$ $e_j$ $\Rightarrow$ i = j, i.e, aucune hyperarête de $H$ ne renferme une autre hyperarête, sinon, $H$ est dit hypergraphe multiple.
\end{definition}

Ainsi, la définition des hypergraphes englobe celle des graphes. En effet, un graphe simple est un hypergraphe simple dont toutes les hyperarêtes sont de cardinalité 2, i.e., $\mid$ $e_i$ $\mid$ = $2$ $\forall$ $e_i$ $\in$ $\xi$.

\begin{property}\label{sperner} \cite{berge1989}\\
Tout hypergraphe simple $H$ d'ordre $n$ vérifie :
\begin{center}
$ \displaystyle { \sum_{e \in \xi}^{}} {\begin{pmatrix}
   n \\
   \mid e \mid
\end{pmatrix}}^{-1}$ $\preceq$ 1.
\end{center}

De plus, le nombre d'arêtes vérifie :
\begin{center}
$\mid \xi \mid$ $\preceq$ ${\begin{pmatrix}
   n \\
   \lfloor n/2 \rfloor
\end{pmatrix}}$
\end{center}
\end{property}

\begin{definition} \textsc{sous-hypergraphe} \cite{berge1989}\\
Soit l'hypergraphe $H = (\mathcal{X}, \xi)$ et $\mathcal{Y}$ $\subseteq$ $\mathcal{X}$, nous appelons sous-hypergraphe de $H$ tout hypergraphe $H^\mathcal{Y}$ = ($\mathcal{Y}$, $\xi^{\mathcal{Y}}$) engendré par $\mathcal{Y}$, tel que $\xi^{\mathcal{Y}}$ = \{ ${e^{\mathcal{Y}}_i} = e_i \cap \mathcal{Y}$ $\mid$ $e_i \in$ $\xi$ et ${e^{\mathcal{Y}}_i}$ $\cap$ $\mathcal{Y}$ $\neq \emptyset$\}

\end{definition}

\begin{definition} \textsc{hypergraphe partiel} \cite{berge1989}\\
Soit l'hypergraphe $H = (\mathcal{X}, \xi)$ et $\xi' \subset \xi$, nous appelons hypergraphe partiel de $H$ tout hypergraphe $H' = (\mathcal{X'}, \xi')$ engendré par $\xi'$, tel que $\mathcal{X'}$ = $\bigcup_{e' \in \xi'} e'$. $H'$ est alors la restriction de l'hypergraphe $H$ à un sous-ensemble d'hyperarêtes $\xi'$ inclus dans $\xi$ et aux sommets inclus dans $\mathcal{X'}$.
\end{definition}

Le rang \emph{r($H$)} d'un hypergraphe $H$ est, le nombre maximum de sommets d'une hyperarête et est défini par \emph{r($H$)} = max\{$\mid$$e_i$$\mid$, $\forall$ $e_i$ $\in$ $\xi$\}. Inversement, l'anti-rang \emph{ar($H$)} désigne le nombre minimum de sommets d'une hyperarête, i.e, \emph{ar($H$)} = min\{$\mid$ $e_i$ $\mid$, $\forall$ $e_i$ $\in$ $\xi$\} \cite{berge1989} \cite{EiterG95}. Trivialement, \emph{ar($H$)} $\leq$ \emph{r($H$)}. Si le rang et l'anti-rang d'un hypergraphe $H$ sont égaux, alors $H$ est dit \emph{uniforme}. Tout hypergraphe uniforme est simple.

\begin{example}
Considérons l'hypergraphe $H$ de la Figure \ref{hypchap1}. Nous avons \emph{r($H$)} = 3 et \emph{ar($H$)} = 1. Par conséquent, $H$ n'est pas uniforme.
\end{example}

De plus, $H$ est dit \emph{n-uniforme} si $H$ est un hypergraphe simple uniforme de rang \emph{n} \cite{EiterG95}. En ce sens, tout graphe est un hypergraphe uniforme de rang égal à $2$.

Un hypergraphe est dit "intersectant" si aucun couple de ses hyperarêtes n'est disjoint, i.e., $\forall e_1$, $e_2 \in \xi$, $e_1$ $\cap$ $e_2$ $\neq \emptyset$ \cite{EiterG95}.

Dans un hypergraphe, deux sommets $x_i$ et $x_j$ sont dits \emph{adjacents} s'il existe une hyperarête $e_i$ qui les contient tous les deux ; deux hyperarêtes $e_i$ et $e_j$ sont dites adjacentes si leur intersection est non vide.\\

Un hypergraphe $H = (\mathcal{X}, \xi)$ peut être représenté par une matrice d'incidence, notée $IM_H$, définie comme suit : \\

\begin{center}
$IM_H[e_i, x_j] \left\{
    \begin{array}{lll}
        = 1 & si & x_j \in e_i\\
     = 0 & sinon &
    \end{array}
\right.$
\end{center}

\begin{definition}
\textsc{Hypergraphe dual}\\
A tout hypergraphe $H = (\mathcal{X}, \xi)$, nous pouvons faire correspondre un hypergraphe $H^*$ = ($\mathcal{X}^*$, $\xi^*$) tel que $\mathcal{X}^*$ = $\xi$ et $\xi^*$ = $\mathcal{X}$. Les sommets $x^*_1$, $x^*_2$, $\ldots$, $x^*_m$ représentent les hyperarêtes de $H$ et les hyperarêtes $e^*_1$, $e^*_2$, $\ldots$, $e^*_n$ représentent les sommets $x_1$, $x_2$, $\ldots$, $x_n$ de $H$, où : \\
\begin{center}
$X_j$ = \{$E_i$ $\mid$ i $\preceq$ m, $e_i$ $\ni$ $x_j$\} \bigskip (j = 1, 2, $\ldots$, n).
\end{center}
On a $X_j$ $\neq$ $\emptyset$, $\bigcup_{j} X_i$ = $\xi$, donc $H^*$ est bien un hypergraphe.
\end{definition}

\begin{figure}
\begin{center}
\includegraphics[scale=0.7]{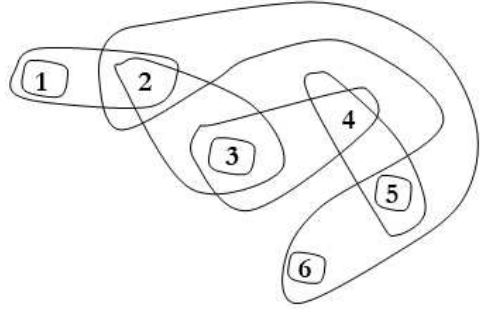}
\caption{Hypergraphe dual}\label{hypdual}
\end{center}
\end{figure}

\begin{example}
Reconsidérons l'hypergraphe $H$ de la Figure \ref{hypchap1}. La Figure \ref{hypdual} illustre l'hypergraphe dual $H^*$ de $H$ tel que $H^*$ = ($\mathcal{X}^*$, $\xi^*$) $\mathcal{X}^* = \{1, 2, 3, 4, 5, 6\}$ et  $\xi^* = \{\{$1$\}, \{$1, 2$\}$, $\{$2, 3$\}$, $\{$3, 4$\}$, \{$3$\} $\{$4, 5$\}$, $\{$2, 5, 6$\}$, $\{$5$\}$, \{$6$\}\}.
\end{example}

$H^*$ est appelé l'\emph{hypergraphe dual} de $H$. La matrice d'incidence $IM_H$ de l'hypergraphe $H$, et la matrice d'incidence $IM_{H^*}$ de l'hypergraphe $H^*$, se déduisent l'une de l'autre par transposition ; on a donc en particulier $(H^*)^*$ = $H$. Si deux sommets $x_i$ et $x_j$ de $H$ sont adjacents, il leur correspond dans $H^*$ des hyperarêtes $e^*_i$ et $e^*_j$ adjacentes ; si deux hyperarêtes $e_i$ et $e_j$ de $H$ sont adjacentes, il leur correspond des sommets $x^*_i$ et $x^*_j$ adjacents de $H^*$.

\begin{definition}\label{TMberge}
\textsc{Traverse minimale}~\cite{berge1989}\\
 Soit un  hypergraphe $H = (\mathcal{X}$, $\xi)$. L'ensemble des traverses de $H$, noté $\gamma_H$, est égal à : $\gamma_H$ = $\{T \subset \mathcal{X}$ $|$ $T$ $\bigcap$ $e_i$ $\neq  \emptyset$, $\forall i = 1, \ldots, \mid \xi \mid\}$.

 Une traverse $ T $ de $\gamma_H$ est dite \emph{minimale} s'il n'existe pas une autre traverse S de $\gamma_H$ incluse dans $T$ :
$\nexists$ $S \in \gamma_H~s.t.~ S\subset T $.

\end{definition}

Nous noterons $\mathcal{M}_H$, l'ensemble des traverses minimales définies sur $H$. \\Dans l'exemple illustratif de la Figure \ref{hypchap1}, l'ensemble $\mathcal{M}_H$ des traverses minimales de l'hypergraphe est : $\{$ \{1, 4, 7\}, \{2, 4, 7\}, \{1, 3, 6, 7\}, \{1, 3, 6, 9\}, \{1, 5, 6, 7\}, \{2, 3, 6, 7\}, \{2, 3, 6, 9\}, \{2, 5, 6, 7\}, \{2, 4, 6, 9\}, \{2, 4, 8, 9\}, \{2, 5, 6, 9\}, \{1, 3, 4, 8, 9\}$\}$.

A partir d'un hypergraphe $H$ = $(\mathcal{X}$, $\xi)$, l'ensemble des traverses minimales $\mathcal{M}_H$ permet la construction de l'hypergraphe transversal, que nous avons noté $H^t$ = $(\mathcal{X}^t$, $\xi^t)$, tel que $\xi^t$ = $\mathcal{M}_H$ et $\mathcal{X}^t = \bigcup \limits_{i=1}^{\mid \xi^t \mid} e^t_i$ $\forall$ $e^t_i$ $\in \xi^t$ \cite{EiterIA}.




\begin{lemma}\cite{berge1989}
Pour tout hypergraphe simple $H$, nous avons $H^{(t)^{(t)}}$ = $H$.

\end{lemma}

\begin{definition}\label{TMberge}
\textsc{Nombre de Transversalité}~\cite{berge1989}\\
 Soit un  hypergraphe $H = (\mathcal{X}$, $\xi)$, le nombre minimum de sommets d'un ensemble transversal est appelé le nombre de transversalité de l'hypergraphe $H$ et est désigné par :

 \begin{center}
 $\tau$($H$) = min \{$|T|$, s.t. $T \in \mathcal{M}_H$\}.
 \end{center}

\end{definition}

Ainsi, dans l'exemple illustratif de la Figure \ref{hypchap1}, le nombre de transversalité de l'hypergraphe $H$ est égal à $3$ car la plus petite traverse minimale de $\mathcal{M}_H$ renferme $3$ sommets.

La détermination d'un nombre de transversalité apparaît dans de nombreux problèmes combinatoires comme la détermination d'un ensemble stable maximum d'un graphe ou encore la détermination d'un ensemble absorbant minimum d'un $1$-graphe \cite{berge1989}.

\section{Problème de l'extraction des traverses minimales}
L'extraction des traverses minimales d'un hypergraphe est un des problèmes les plus importants en théorie des hypergraphes. C'est un problème algorithmique central et particulièrement difficile et la question de sa complexité exacte reste toujours ouverte. Plusieurs travaux se sont attachés à proposer diverses méthodes pour le traiter \cite{berge1989} \cite{KavvadiasS05} \cite{Boros03anefficient}. Fredman et Khachiyan ont proposé un algorithme avec une complexité quasi-polynomiale de $N(o^{log N})$ où $N$ représente la taille de l'entrée et de la sortie \cite{Fredkach96}. Ce résultat de Fredman et Khachiyan nous donne un algorithme d'extraction des traverses minimales dont la complexité est bornée par $\mid \mathcal{M}_H \mid$ $\times$ ${(\mid \mathcal{M}_H \mid + \mid H \mid)}^{o(\mid \mathcal{M}_H \mid + \mid H \mid)}$ \cite{thesearnaud}. Ce résultat relance le débat sur le fait que ce problème soit coNP-complet puisqu'à moins que tout problème coNP-complet admette un algorithme quasi-polynomial, le problème de l'extraction des traverses minimales n'est pas coNP-complet.

Trouver une traverse minimale d'un hypergraphe est une tâche aisée mais calculer l'ensemble de toutes les traverses minimales pose plusieurs problèmes dans la mesure où le nombre de sous-ensembles de sommets à tester est très important. Les travaux antérieurs, pour faire sauter les différents verrous scientifiques que posait l'extraction des traverses minimales d'un hypergraphe, se sont attachés à réduire l'espace de recherche. Néanmoins, le coût du calcul reste substantiellement élevé et les algorithmes existants se sont heurtés à des temps d'exécution conséquents et à l'incapacité de traitement lorsque le nombre de transversalité de l'hypergraphe d'entrée est grand.

Comme nous l'avons mentionné dans la section précédente, le problème de l'extraction des traverses minimales à partir d'un hypergraphe $H$ est équivalent à celui de la construction de l'hypergraphe transversal à $H$. Formellement, nous définissions ce problème comme suit: \\

\fbox{
   \begin{minipage}{0.9\textwidth}
      \textbf{Entrée:} Hypergraphe simple $H$ = $(\mathcal{X}$, $\xi)$.\\
      \textbf{Sortie:} Hypergraphe transversal $H^t$ = $(\mathcal{X}^t$, $\xi^t)$.
   \end{minipage}%
}
\\ \\
Même pour des hypergraphes simples, le nombre de traverses minimales d'un hypergraphe peut être exponentiel \cite{thesehagen}, comme le montre l'exemple \ref{expnb}.

\begin{example}\label{expnb}
Soit $H = (\mathcal{X}, \xi)$ un hypergraphe tel que $\mathcal{X}$ = ($x_1$, $x_2$, \ldots, $x_{2n}$) et $\xi = (\{$$x_1$, $x_2$$\}$, $\{$$x_3$, $x_4$$\}$, \ldots, $\{$$x_{2n-1}$, $x_{2n}$$\}$). $H$ est de taille $2n$ mais renferme $2^n$ traverses minimales.
\end{example}

La complexité des approches proposées dans la littérature, et décrites dans le chapitre suivant, est analysée en termes de la taille d'entrée et de sortie. Plus concrètement, si \emph{n} = $\mid \mathcal{X} \mid$, \emph{m} = $\mid \xi \mid$ et \emph{m'} = $\mid \mathcal{M}_H \mid$, nous disons qu'un algorithme d'extraction des traverses minimales est polynomial en la taille de l'entrée et de la sortie $N$ si sa complexité peut être bornée, de manière polynomiale, par $N$, qui désigne une fonction de \emph{n}, \emph{m} et \emph{m'}. En outre, un algorithme est incrémental s'il énumère une par une toutes traverses minimales de l'hypergraphe d'entrée de telle sorte que le temps nécessaire pour délivrer en sortie une nouvelle transverse minimale est polynomiale en \emph{n}, \emph{m} et \emph{k}, \emph{k} étant la taille de l'hypergraphe transversal.

La notion d'algorithme incrémental a ouvert la voie à une autre approche consistant à ne générer qu'un sous-ensemble de traverses minimales, i.e., une liste partielle de traverses minimales, à partir de l'hypergraphe d'entrée. Formellement, ce problème est défini comme suit :\\

\fbox{
   \begin{minipage}{0.9\textwidth}
      \textbf{Entrée:} Hypergraphe simple $H$ = $(\mathcal{X}$, $\xi)$ et un sous-ensemble de traverses minimales $S$ $\subseteq$ $\mathcal{M}_H$.\\
      \textbf{Sortie:} Vrai si $S$ $\subseteq$ $\mathcal{M}_{H}$, sinon retourner une traverse minimale de $\mathcal{M}_{H}$ $\backslash$ $S$.
   \end{minipage}%
}
\\ \\
Un troisième problème a été défini par \cite{Bioch} et qui se résume à vérifier si deux hypergraphes sont transversaux l'un par rapport à l'autre.
\\

\fbox{
   \begin{minipage}{0.9\textwidth}
      \textbf{Entrée:} Deux hypergraphes simples $H$ = $(\mathcal{X}$, $\xi)$ et $H'$ = $(\mathcal{X'}$, $\xi')$.\\
      \textbf{Sortie:} Vrai si $H'$ = $H^t$, Faux sinon.
   \end{minipage}%
}
\\ \\
Les trois problèmes sont liés et divers algorithmes ont été proposés pour les résoudre.

\section{Domaines d'application}\label{DomAppli}
L'intérêt pour l'extraction des traverses minimales s'est accru, ces dernières années, en raison de la diversité des domaines d'application où le recours aux traverses minimales peut constituer une solution. Le large éventail des domaines d'application, comme le résume la Figure \ref{appli} \cite{thesehagen}, donne ainsi une importance plus grande aux traverses minimales et motive l'intérêt qu'elles suscitent. Dans ce qui suit, nous en donnerons un aperçu et nous citerons les problèmes les plus connus, où les traverses minimales sont applicables.

\begin{figure}[h]
\begin{center}
\includegraphics[scale=0.5]{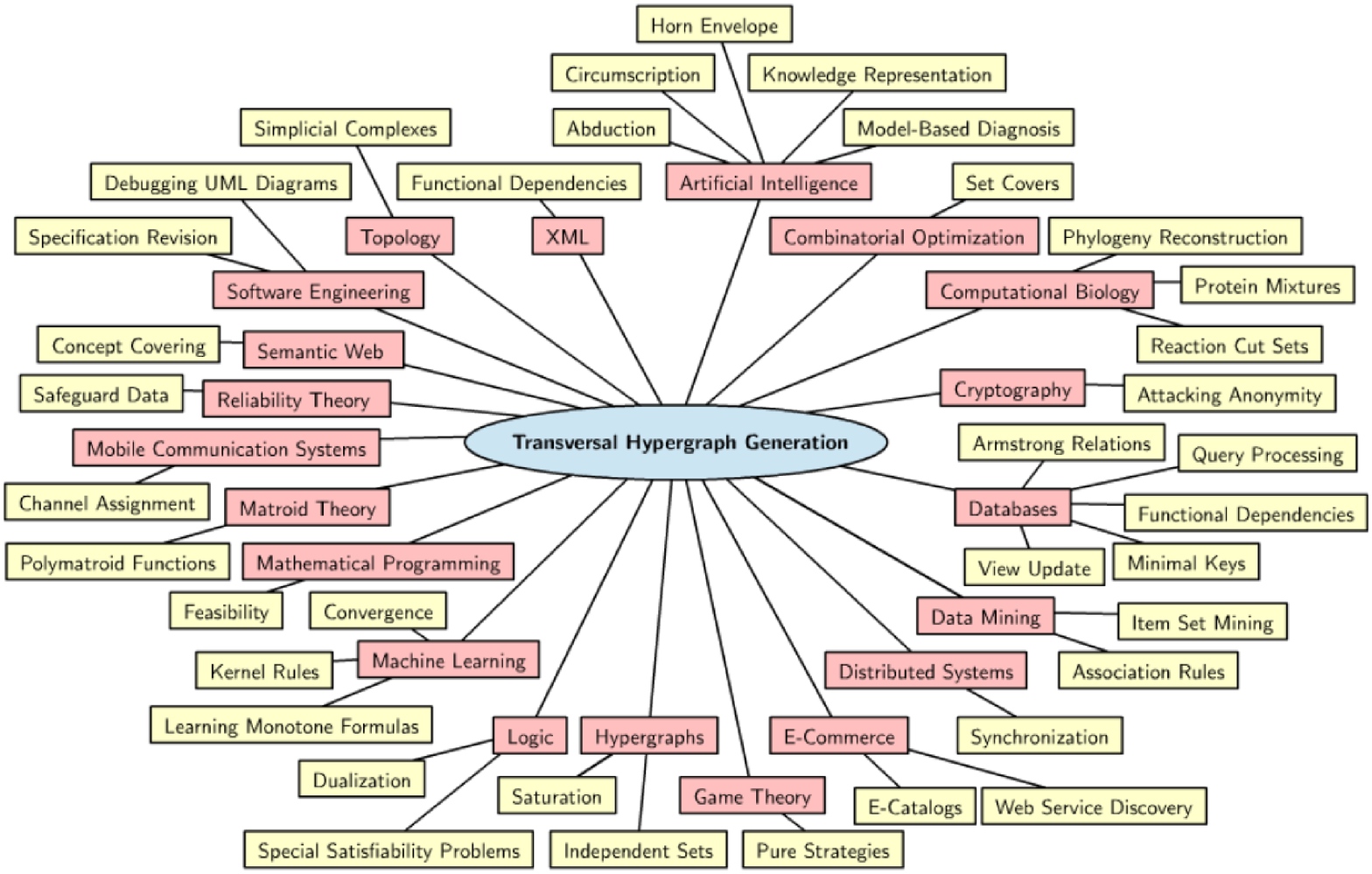}
\caption{Domaines d'application des traverses minimales \cite{thesehagen}}\label{appli}
\end{center}
\end{figure}

\subsection{Bases de données}
Plusieurs travaux se sont attachés à appliquer les traverses minimales pour résoudre des verrous scientifiques dans le domaine des bases de données. Jouant un rôle important dans l'identification, de façon minimale, des n-uplets que renferment les relations, l'identification des clés est fortement liée au problème du calcul des traverses minimales comme l'ont démontré les travaux de \cite{Janos99} et de \cite{DucVu}. Étant donné une relation et un ensemble de clés, décider de l'existence d'une autre clé est un problème, équivalent à celui de la recherche des traverses minimales. Les dépendances d'inclusion \cite{map03}, qui sont une généralisation des clés étrangères dans un modèle relationnel, peuvent ainsi être déduites en adaptant les techniques de calcul des traverses minimales d'un hypergraphe. Celles-ci peuvent, par ailleurs, présenter des solutions aux problèmes de réécriture des requêtes, d'exécution des requêtes et d'actualisation des vues. Ces dernières, dont le rôle est très important dans la présentation des données à partir des bases, peuvent en effet être gérées en se basant sur les traverses minimales. L'inférence des dépendances fonctionnelles représente aussi un domaine d'application fort intéressant des traverses minimales comme le montre les travaux de \cite{Mannila1994} et \cite{Lopes2000} et comme nous le verrons dans le chapitre $4$.

\subsection{Logique}
En logique, une \emph{clause} est une disjonction de littéraux, qui sont des variables booléennes ou leurs négations, alors qu'un \emph{terme} est une conjonction. Une formule est sous sa forme normale \emph{disjonctive} (FND) (resp. \emph{conjonctive} (FNC)) si c'est une une disjonction de termes (resp. une conjonction de clauses). La détermination de la dualité FNC/FND est équivalente au problème de calcul des traverses minimales d'un hypergraphe. En effet, le problème, connu en logique, de la dualisation où il s'agit de calculer la forme duale normale disjonctive (\textsc{fnd}) monotone et irrédondante à partir d'une forme normale disjonctive du même type est équivalent au calcul d'un sous-ensemble des traverses minimales d'un hypergraphe.

\subsection{Intelligence artificielle}
Plusieurs problématiques en intelligence artificielle ont un lien très fort avec les traverses minimales, à commencer par l'abduction. Définie par Peirce, l'abduction est un mode de raisonnement par lequel des faits utiles sont inventés (contrairement à l'induction qui consiste à inventer des théories) et est utilisée dans deux acceptations différentes. Plus formellement, à partir de l'ensemble caractéristique d'une théorie de Horn $\Sigma$, d'un littéral \emph{q} et d'un sous-ensemble \emph{A} de tous les littéraux, il s'agirait de trouver toutes les interprétations possibles pour \emph{q} par rapport à A. Dans ce cas, une théorie logique est un ensemble de formules. C'est une Horn s'il s'agit d'un ensemble des clauses ayant, chacune, au plus un littéral positif. Les travaux de Eiter et \emph{al.} \cite{EiterIA} ont montré que ce problème de l'abduction est équivalent à une adaptation du calcul des traverses minimales d'un hypergraphe.

\subsection{E-commerce}
Que ce soit pour des problématiques telles que la recherche du meilleur recouvrement, dont les applications sont nombreuses, la réécriture des requêtes dans les e-catalogues ou la découverte des web services, les techniques d'extraction des traverses minimales peuvent jouer un rôle important en e-commerce, comme le montre les travaux de \cite{Boualem05}. Ainsi, à titre d'exemple, l'une des tâches les plus importantes en web services est de trouver automatiquement les services qui répondent à des contraintes d'utilisation spécifiques aux utilisateurs. Ce problème est singulièrement similaire à celui de la recherche d'un ensemble de couverture dans un contexte de contraintes et pour lequel l'application des techniques de calcul des traverses minimales s'avère judicieuse.

\subsection{Fouille de données}
Les liens entre certaines problématiques de la fouille de données \cite{ayouni2011extracting, ferjani2012formal} et les traverses minimales sont nombreuses. C'est d'ailleurs le domaine où l'application des traverses minimales est la plus intéressante, comme dans la génération des règles associatives, des itemsets fermés, des itemsets fréquents ou encore l'analyse formelle de concepts \cite{yahia2004revisiting, hamrouni2010generalization, brahmi2012omc}. \\
Les règles associatives sont de la forme "si \emph{x} est présent dans une transaction alors il y a 95\% de chance que \emph{y} soit aussi présent". Elles sont cruciales au cours du processus d'extraction des connaissances. Une des étapes les plus importantes dans la génération des règles associatives est le calcul des ensembles fréquents et infréquents. Or, calculer les ensembles maximaux \emph{k}-fréquents ou les ensembles minimaux \emph{k}-infréquents est équivalent à la construction d'un hypergraphe transversal comme l'ont démontré les travaux de Boros et \emph{al.} \cite{BorosMax} et, Mannila et Toivonen \cite{Mannila97}.
Par ailleurs, les règles associatives ayant comme prémisses les générateurs minimaux sont les règles les plus intéressantes en fouille de données. Les générateurs minimaux d'un itemset fermé $F$ étant les plus petits itemsets ayant cette même fermeture $F$, leur calcul peut être optimisé en utilisant les traverses minimales d'un hypergraphe selon les travaux de \cite{thesegarriga} et de \cite{Pfaltz02}.
Les treillis de concepts qui jouent un rôle important dans l'extraction de connaissances \cite{Pfaltz02} et la génération des règles à partir de bases de données relationnelles \cite{YahiaN04, bouzouita2006garc, gasmi2007extraction, bouker2012ranking} peuvent aussi bénéficier des techniques d'extraction des traverses minimales en raison de leur relation avec les itemsets fermés.

\section{Conclusion}
Nous avons introduit, dans ce chapitre, les notions-clés de la théorie des hypergraphes tout en mettant en exergue les verrous scientifiques que posent le problème de l'extraction des traverses minimales d'un hypergraphe. Le survol des domaines d'application des traverses minimales démontre l'intérêt, de plus en plus, croissant pour les traverses minimales et, dans la littérature, plusieurs algorithmes dédiés à leur calcul ont été proposés. Dans le chapitre suivant, nous présentons un état de l'art détaillé de ces algorithmes, qui reposent sur des approches différentes.

\chapter{État de l'art}\label{chap2}
\section{Introduction}
Plusieurs auteurs se sont intéressés au problème de l'extraction des traverses minimales d'un hypergraphe. Dans ce chapitre, nous présentons un état de l'art détaillé de ces différentes approches, en mettant en exergue leurs points forts et leurs limites. Le nombre de traverses minimales d'un hypergraphe pouvant être exponentiel en la taille de l'hypergraphe, la question de la mise en place d'un algorithme résolvant le problème de l'extraction des traverses minimales d'un hypergraphe $H$ avec une complexité polynomiale en $\mid$$H$$\mid$ reste néanmoins ouverte. Une étude comparative des différents algorithmes existants nous permet de situer nos contributions par rapport à ces travaux et mettre en lumière l'intérêt des différentes approches que nous proposons dans les chapitres suivants. 
\section{Algorithme de \textsc{Berge} \cite{berge1989}}
\textsc{Berge} est le premier à s'être intéressé au problème du calcul des traverses minimales et à avoir proposé un algorithme pour le résoudre. Cet algorithme, dont le principe est simple, commence par calculer l'ensemble des traverses minimales de la première  hyperarête, qui est équivalent à l'ensemble des sommets contenus dans cette dernière. Ensuite, il met à jour cet ensemble des traverses minimales en ajoutant les autres hyperarêtes, une à une, de manière incrémentale. Ainsi, l'algorithme de Berge construit des hypergraphes partiels au fur et à mesure qu'il ajoute des hyperarêtes. Néanmoins, l'algorithme a toujours besoin de stocker les traverses minimales intermédiaires avant de passer à l'étape suivante consistant à ajouter une nouvelle hyperarête. \\
Formellement, l'algorithme de Berge repose sur la formule suivante \cite{berge1989}: à partir d'un hyperpgraphe partiel $H_i$ = ($\mathcal{X}$, $\xi$) tel que $\mathcal{X}$ = \{$x_1$, $x_2$, \ldots, $x_n$\} et $\xi$ = \{$e_1$, $e_2$, \ldots, $e_i$\}, l'ensemble des traverses minimales de $\{H_{i-1} \cup {e_{i}}\} = \{min\{T \times \{x\}\}$ tel que $T \in \mathcal{M}_{H_{i-1}}$  et  $x \in e_{i}\}$. L'opérateur "$\times$" désigne le produit cartésien tel que $A \times B$ comporte toutes les pairs ($a$,$b$) tel que $a \in A$ et $b \in B$.

\begin{algorithm}[!h]
{
\LinesNumbered
\Entree{$H$ = ($\mathcal{X}$, $\xi$): Hypergraphe}
\Sortie{$\mathcal{M}_H$ : ensemble des traverses minimales de $H$}

\Deb
{
$\mathcal{M}_{H_{1}}$ = \{\{$x$\} $\mid$ $x \in \xi_1$\};\\
\Pour{$i$ = $2 \rightarrow $ $\mid \xi \mid$}{
    $\mathcal{M}_{H_{i}}$ = Min \{$T$ $\times$ \{$x$\} tel que $T$ $\in$ $\mathcal{M}_{H_{i-1}}$, $x$ $\in$ $e_i$\} }

$\mathcal{M}_H$ = $\mathcal{M}_{H_{\mid \xi \mid}}$;\\
\Retour{$\mathcal{M}_H$}

}
}
 \caption{L'algorithme de \textsc{Berge} \cite{thesehagen}}
  \label{algoBerge}
\end{algorithm}

La complexité de cet algorithme, dont le pseudo-code est donné par l'Algorithme \ref{algoBerge} est exponentielle en la taille de l'entrée et de la sortie \cite{thesehagen}. Ceci s'explique par la nécessité de stocker toutes les traverses intermédiaires vu que l'ensemble des traverses minimales n'est généré qu'après l'insertion de la dernière hyperarête. Ceci rend l'algorithme de \textsc{Berge} impraticable sur des hypergraphes de grande taille.

Récemment, Boros \emph{et al.} ont prouvé que le temps de traitement de l'algorithme de Berge a une borne supérieure subexponentielle de $N^{\sqrt{N}}$ \cite{Boros2008}.
\section{Améliorations de l'algorithme de \textsc{Berge}}
Plusieurs chercheurs ont cherché à améliorer l'algorithme de \textsc{Berge}. Parmi les améliorations les plus connues proposées récemment figurent celles introduites par Dong \emph{et al.} \cite{dongg}, Kavvadias et Stavropoulos \cite{KavvadiasS05} et Bailey \emph{et al.} \cite{BMA03}.

\begin{algorithm}[!h]
{
\LinesNumbered
\Entree{$H$ = ($\mathcal{X}$, $\xi$): Hypergraphe}
\Sortie{$\mathcal{M}_H$ : ensemble des traverses minimales de $H$}

\Deb
{
$\mathcal{M}_{H_{1}}$ = \{\{$x$\} $\mid$ $x \in \xi_1$\};\\
\Pour{$i$ = $2 \rightarrow $ $\mid \xi \mid$}{
$T_g$ = \{$t \in \mathcal{M}_{H_{i-1}} \mid t \cap e_i \neq \emptyset$\};\\
$e_i^{cov}$ = \{$x \in e_i \mid$ \{$x$\} $\in T_g$\};\\
$\mathcal{M'}_{H_{i-1}}$ = $\mathcal{M}_{H_{i-1}} \backslash T_g$;\\
$e'_i$ = $e_i \backslash e_i^{cov}$;\\

\PourCh{$t' \in \mathcal{M'}_{H_{i-1}}$ trié par ordre croissant de cardinalité}{
\PourCh{$x \in e'_i$}{\Si{$t' \cup$ \{$x$\} n'est le sur-ensemble d'aucun élément de $T_g$}{$T_g$ = $T_g \cup$ \{$t' \cup$ \{$x$\}\} }
}
}

$\mathcal{M}_{H_{i}}$ = $T_g$;\\
}
$\mathcal{M}_{H}$ = $\mathcal{M}_{H_{\mid \xi \mid}}$;\\
\Retour{$\mathcal{M}_{H}$}
}
}
  \caption{L'algorithme de Dong et Li \cite{thesehagen}}
  \label{Dong}
\end{algorithm}

\subsection{Algorithme de Dong et Li \cite{dongg}}
C'est en s'inspirant de l'extraction des itemsets émergeants en fouille de données que Dong et Li ont proposé une solution, dont le pseudo-code est donné par l'Algorithme \ref{Dong}.

L'algorithme de Dong et Li a été évalué expérimentalement sur de nombreux jeux de données. Néanmoins, les auteurs n'ont pas effectué une analyse théorique de la complexité en temps de traitement de leur algorithme. Cependant, cette adaptation de l'algorithme initial s'est avérée très fructueuse. La principale amélioration de Dong et Li par rapport à l'algorithme de Berge réside dans l'optimisation réalisée lors du calcul de $\mathcal{M}_{H_{i-1}} \times \{\{x\} \mid x \in e_i\}$, et qui consiste à considérer uniquement les traverses qui intersectent la nouvelle hyperarête traitée et, aussi, à ne prendre que les sommets de $\xi_i$ qui n'appartiennent pas déjà aux traverses minimales déjà identifiées.

\subsection{Algorithme de Kavvadias et Stavropoulos \cite{KavvadiasS05}}\label{KS05}
L'un des inconvénients majeurs de l'algorithme de \textsc{Berge}, observé par Kavvadias et Stavropoulos, est la consommation excessive en mémoire. Dans la mesure où la minimalité des nouvelles traverses calculées doit être testée, les traverses minimales intermédiaires doivent aussi être stockées en mémoire. L'algorithme de Kavvadias et Stavropoulos tente de surmonter ce problème de consommation mémoire, en utilisant deux techniques. La première introduit la notion de "\emph{sommets généralisés}" selon la Définition \ref{gener}.

\begin{definition}\label{gener}
Soit $H = (\mathcal{X}$, $\xi)$ un hypergraphe. L'ensemble $X$ $\subseteq$ $\mathcal{X}$ est un ensemble de sommets généralisés de $H$ si tous les sommets de $X$ appartiennent aux mêmes hyperarêtes de $\xi$.

\end{definition}
Le pseudo-code de l'algorithme de Kavvadias et Stavropoulos est décrit par l'Algorithme \ref{KS}. Le principe est le suivant.
En ajoutant une hyperarête $e_i$, l'algorithme met à jour l'ensemble des sommets généralisés avant de considérer les éléments de $M_{H_{i-1}^g}$ et les sommets constituant $\xi_{i}$ comme les ensembles de sommets généralisés du niveau $i$. $H_{i-1}^g$ étant l'hypergraphe partiel composé uniquement des sommets généralisés calculés au niveau $i-1$, les traverses minimales de l'hypergraphe $H_{i}^g$, $M_{H_{i}^g}$, sont ensuite calculées selon la formule de Berge, i.e., en effectuant le produit cartésien entre $M_{H_{i-1}^g}$ et les sommets généralisés de $\xi_{i}^g$, et en testant la minimalité de ces traverses candidates.

La seconde technique introduite par Kavvadias et Stavropoulos pour diminuer la consommation mémoire élevée revient à adopter une stratégie de recherche en profondeur d'abord. Berge utilisait une forme de parcours en largeur d'abord à travers la construction d'un "arbre" de traverses minimales. Au \emph{i}-ème niveau de l'arbre, les noeuds sont des traverses minimales de l'hypergraphe partiel $H_i$. Les descendants d'une traverse minimale $T$, du niveau \emph{i}, sont les traverses minimales de l'hypergraphe $H_{i+1}$ incluant $T$. Le parcours de cet "arbre" est très coûteux puisque les traverses minimales d'un hypergraphe $H$ sont retrouvées dans le dernier niveau de l'arbre. De plus, certains noeuds sont "visités" plusieurs fois parce qu'ils peuvent avoir plusieurs parents.

Pour remédier à ce problème, Kavvadias et Stavropoulos utilisent une stratégie en profondeur d'abord et introduisent la notion de "\emph{sommets appropriés}" pour vérifier la minimalité des traverses générées. Ceci permet à l'algorithme de réduire considérablement le stockage en mémoire durant le calcul des traverses minimales d'un hypergraphe.

\begin{definition}
Soit un hypergraphe $H = (\mathcal{X}$, $\xi)$ et soit $T$ une traverse minimale de l'hypergraphe partiel $H_{i-1}$ de $H$. Un ensemble de sommets généralisés $X$ $\subseteq$ $\mathcal{X}$ $\backslash$ $T$, au niveau $i$ est un ensemble de \textit{sommets appropriés} pour $T$ si aucun sous-ensemble de $T$ $\cup$ $X$, excepté $X$, ne peut être supprimé sans que les sommets restants ne représentent plus une traverse.
\end{definition}

\begin{algorithm}[!h]
{
\LinesNumbered
\Entree{$H$ = ($\mathcal{X}$, $\xi$): Hypergraphe}
\Sortie{$\mathcal{M}_H$ : ensemble des traverses minimales de $H$}

\Deb
{
\Pour{$k$ = $0$ $\rightarrow$ $\mid \xi \mid$}{Ajouter\_hyperarête($\xi_{k+1}$);\\
Mettre à jour les sommets généralisés;\\
Considérer $M_{H_k^g}$ et $\xi_{k+1}$ comme étant des sommets généralisés du niveau $k+1$;\\
Calculer $M_{H_{k+1}^g}$ = $Min\{M_{H_k^g} \times \{\{x_i\}$ :$ x_i \in \xi_{k+1}^g\}\}$;\\
}
Déduire $M_{H}$ à partir de $M_{H_{\mid \xi \mid}}$;\\
\Retour{$M_{H}$}
}
}
  \caption{L'algorithme de Kavvadias et Stavropoulos \cite{KavvadiasS05}}
  \label{KS}
\end{algorithm}

L'algorithme de Kavvadias et Stavropoulos n'est pas polynomial en la taille de la sortie. Son temps de traitement est de l'ordre de $N^{\Omega(log log N)}$, \emph{N} désignant la taille de l'entrée et de sortie \cite{thesehagen}. Cet algorithme est l'un des plus performants, en termes de temps de traitement. Adoptant une stratégie en profondeur, l'algorithme consomme, par ailleurs, très peu de mémoire vive. Ce qui représente un avantage non négligeable.

\subsection{Algorithme de Bailey \emph{et al.} \cite{BMA03}}
Pour traiter les hypergraphes de grande taille, Bailey \emph{et al.} ont exploité les bonnes performances de l'algorithme de Dong et Li sur les hypergraphes renfermant des hyperarêtes de petite taille. L'algorithme de Bailey \emph{et al.}, dont le pseudo-code est décrit par l'Algorithme \ref{bailey}, prend en entrée un hypergraphe et comporte un pré-traitement récursif.

\begin{algorithm}[!h]
{
\LinesNumbered

\Entree{$H$ = $(\mathcal{X}$, $\xi)$: Hypergraphe}
\Sortie{$\mathcal{M}_H$ : ensemble des traverses minimales de $H$}

\Deb
{
$X$ = $\mathcal{X}$;\\
Ordonner($\mathcal{X}$);\\
\Pour{$i$ = $1 \rightarrow \mid \mathcal{X} \mid$}{
$\xi_{part}$ = $\emptyset$;\\
$X$ = $X$ $\backslash$ {$x_i$}\\

\PourCh{$e$ $\in$ $\xi$}{ \Si{$x_i$ $\not\in$ $e$}{$\xi_{part}$ = $\xi_{part}$ $\cup$ \{e $\backslash$ $X$\};}}
$\mathcal{X}_{part}$ = $\mathcal{X}_{part}$ $\cup$ {$x_i$};\\
\eSi{$\mid$ $\xi_{part}$ $\mid$ $\geq$ 2 and Volume($\xi_{part}$) $\geq$ 50}{Algorithme \ref{bailey} ($\mathcal{X}_{part}$, $\xi_{part}$);}
{$\mathcal{M}_{\xi_{part}}$ = Algorithme \ref{Dong}($\xi_{part}$);\\
$\mathcal{M}_H$ = min($\mathcal{M}_H$ $\cup$ ($\mathcal{M}_{\xi_{part}}$ $\times$ $\mathcal{X}_{part}$));}
$\mathcal{X}_{part}$ = $\mathcal{X}_{part}$ $\backslash$ \{$x_i$\}
}
\Retour{$Tr$}

}

}

  \caption{L'algorithme de Bailey \emph{et al.} \cite{thesehagen}}
  \label{bailey}
\end{algorithm}

A partir d'un sommet ou d'un ensemble de sommets $X_{part}$ apparaissant dans le plus petit nombre d'hyperarêtes dans l'hypergraphe d'entrée, ce pré-traitement vise à construire un sous-ensemble d'hyperarêtes $\xi_{part}$ ne contenant pas ces sommets $X_{part}$. Si $\mid \xi_{part} \mid$ est elevée ($\geq 2$) et si son volume, fonction de la cardinalité moyenne de ces hyperarêtes est aussi elevée ($\geq 50$) alors l'algorithme de Bailey \emph{et al.} est appelé de manière récursive. Sinon, les traverses minimales de $\xi_{part}$ sont déterminées par l'algorithme de Dong et Li. Les traverses minimales de l'hypergraphe d'entrée sont ensuite déduites par la méthode de Berge via un produit cartésien entre les traverses minimales de l'hypergraphe constitué par les hyperarêtes $\xi_{part}$ et $X_{part}$, conjugué à un test de la minimalité. Les expérimentations menées par les auteurs ont montré l'efficacité de leur algorithme, sur un type particulier d'hypergraphes, par rapport aux deux algorithmes considérés, i.e., celui de Fredman et Kachiyan, et celui de Kavvadias et Stavropoulos (une version antérieure à celle présentée dans la section \ref{KS05}) .

\section{Algorithme de Fredman et Kachiyan \cite{Fredkach96}}
En $1996$, Fredman et Kachiyan ont proposé deux algorithmes pour le calcul des traverses minimales d'un hypergraphe, l'algorithme FK-A et sa version optimisée, appelée l'algorithme FK-B. Ce dernier possède la meilleure complexité théorique connue à ce jour et qui est de $N(o^{log N})$ où N représente la taille de l'entrée et de sortie \cite{Fredkach96}. Les auteurs motivent le calcul de traverses minimales comme la solution au problème de la dualisation des formules booléennes monotones et c'est cette approche intuitive que nous reprenons.
\'{E}tant donnée une formule $f(x)$ = $f(x_1, x_2, \ldots, x_n)$ sous forme normale conjonctive, il s'agit de calculer la formule duale correspondante $f^d(x)$ = $\bar{f}(\bar{x})$ = $\bar{f}(\bar{x_1}, \bar{x_2}, \ldots, \bar{x_n})$  sous forme normale conjonctive également. Pour cela, on obtient aisément $f^d$ sous forme normale disjonctive en remplaçant chaque conjonction de $f$ par une disjonction et vice-versa. Pour calculer la forme normale conjonctive de la formule duale, il s'agit finalement de développer la forme normale disjonctive pour constituer les classes de $f^d$. Pour cela, on prendra un littéral dans chaque terme de $\bar{f}$ pour constituer une classe. Des simplifications apparaissent si l'on prend plusieurs fois le même littéral. Pour calculer les traverses minimales d'un hypergraphe, chaque traverse est construite en prenant un item dans chaque terme de $\bar{f}$ : le résultat obtenu est identique à celui fourni par la dualisation des formules booléennes monotones. Ce problème et celui du calcul des traverses minimales d'un hypergraphe sont alors parfaitement équivalents.

Par exemple, soit $f(x)$ = ($x_1 \vee x_2$) $\wedge$ ($x_1 \vee x_2 \vee x_3$) $\wedge$ ($x_1 \vee x_2 \vee x_4$) $\wedge$ ($x_2 \vee x_3 \vee x_4$) $\wedge$ ($x_1 \vee x_2 \vee x_3 \vee x_4$). La formule duale correspondante, obtenue en échangeant chaque conjonction par une disjonction et vice-versa, est $f^d$($x$) = ($x_1 \wedge x_2$) $\vee$ ($x_1 \wedge x_2 \wedge x_3$) $\vee$ ($x_1 \wedge x_2 \wedge x_4$) $\vee$ ($x_2 \wedge x_3 \wedge x_4$) $\vee$ ($x_1 \wedge x_2 \wedge x_3 \wedge x_4$). Si nous développons scrupuleusement cette dernière expression pour la transformer en forme normale conjonctive, on obtient la série de clauses suivantes : $f^d$($x$) = ($x_1 \vee x_1 \vee x_1 \vee x_2 \vee x_1$) $\wedge$ ($x_1 \vee x_1 \vee x_1 \vee x_2 \vee x_2$) $\wedge$ $\ldots$ $\wedge$ ($x_2 \vee x_3 \vee x_4 \vee x_4 \vee x_4$). Ceci donne alors 216 clauses, dont il n'en restera que trois, après les simplifications : $f^d$($x$) = $x_2$ $\wedge$ ($x_1 \vee x_3$) $\wedge$ ($x_1 \vee x_4$).

La solution proposée par les auteurs revient à determiner, de façon incrémentale, si deux formules $f$ et $g$ ne sont pas duales, i.e., $f$($x$) = $\bar{g}$($\bar{x}$).

\begin{algorithm}[!h]
{
\LinesNumbered
\Entree{Deux formules monotones sous la Forme Normale Disjonctive $f$ et $g$}
\Sortie{Le Dual de $f$ et $g$}
\Deb
{
Factoriser $f$ et $g$\\
Vérifier que $f$ et $g$ sont des formes mutuellement duales et que le problème peut se résoudre en un temps polynomial.\\
\Si{$\mid F \mid$ $\mid G \mid$ $\leq$ 1}{le dual de $f$ et $g$ est calculé en $O(1)$.}
\Si{$\mid F \mid$ $\mid G \mid$ $\geq$ 2}{Trouver une variable $x_i$ commune à $f$ et $g$ tel que Fréquence($x_i$) $\geq$ $1/log(\mid F \mid + \mid G \mid)$\\
$f$ = $x_i f_0 \vee f_1$;\\
$g$ = $x_i g_0 \vee g_1$;\\
FK($f_1$, $g_0 \vee g_1$);\\
FK($g_1$, $f_0 \vee f_1$);
}
}
}
  \caption{L'algorithme FK-A \cite{Fredkach96}}
  \label{FK}
\end{algorithm}
La vérification de la dualité, dans l'algorithme de Fredman et Khachiyan dont le pseudo-code est donné par l'Algorithme \ref{FK}, et la mise en évidence d'un disqualifieur sont effectuées grâce à la propriété suivante : en factorisant $f$ et $g$ selon une variable $x_i$, les auteurs font apparaître des formules plus courtes, $f_0, f_1, g_0$ et $g_1$ qui ne contiennent pas $x_i$. On obtient ainsi $f$($x$) = ($x_i$ $\wedge$ $f_0$($y$)) $\vee$ $f_1$($y$) puis $g$($x$) = ($x_i$ $\wedge$ $g_0$($y$)) $\vee$ $g_1$($y$) ($y$ ne contient pas le littéral $x_i$). $f$ et $g$ sont duales si et seulement si $f_1$ et $g_0 \vee g_1$ le sont, ainsi que $f_0 \vee f_1$ et $g_1$. La taille du problème est ainsi réduite et permet d'appliquer récursivement ce procédé. Néanmoins, cette méthode est peu adaptée au calcul des traverses minimales de longueur bornée.

\section{Algorithme \textsc{MtMiner} \cite{HebertBC07}}
L'algorithme proposé par Hébert \emph{et al.} consiste à exploiter les travaux réalisés, dans la littérature, sur l'extraction de motifs \cite{HebertBC07}. Les auteurs ont réutilisé le principe des algorithmes par niveaux pour calculer les traverses minimales d'un hypergraphe. Cette approche repose donc sur le fait que les bases de données et les hypergraphes peuvent se représenter de la même manière, i.e, sous la forme d'une matrice booléenne où les sommets correspondent aux motifs et les hyperarêtes aux objets.

Par le biais de la correspondance de Galois qui relie les ensembles de motifs et les ensembles d'objets, un parallèle est établi entre l'extraction des motifs et l'extraction des traverses minimales. L'extension de cette nouvelle connexion permet de définir des classes d'équivalence, pour un hypergrpahe, de façon analogue aux classes d'équivalence utilisées en fouille de motifs, selon la définition \ref{equiMT} \cite{thCline}.

Ces classes regroupent les ensembles de sommets appartenant aux mêmes hyperarêtes de l'hypergraphe d'entrée $H$. Le nombre d'hyperarêtes non recouvertes par un ensemble de sommets est appelé \emph{fréquence} et correspond alors au nombre d'occurrences (support disjonctif) d'un motif en fouille de données. Les traverses de $H$ sont, dans ce cas, les ensembles de sommets ayant une fréquence nulle. En utilisant les propriétés de la fouille de données, les auteurs ont prouvé que les traverses minimales de $H$ sont les générateurs minimaux de fréquence nulle.

\begin{definition}\label{equiMT}
La classe d'équivalence d'un ensemble de sommets $X$ $\subseteq$ $\mathcal{X}$ est notée $\mathcal{R}_{gH}$ ($X$) et est définie comme suit : \\
$\mathcal{R}_{gH}$ ($X$) = \{$X'$ $\in$ $\mathcal{X}$ $\mid$ $gH$($X'$) = $gH$($X$)\} où $gH$($X$) est l'ensemble des hyperarêtes qui ne contient aucun sommet de $X$.

\end{definition}

L'algorithme \textsc{MtMiner} adopte deux stratégies d'élagage lors du parcours par niveau des candidats dans le treillis généré. La première repose sur la propriété d'anti-monotonie de la minimalité dans les classes d'équivalence, selon laquelle si un ensemble de sommets ne constitue pas un générateur minimal alors l'espace de recherche généré à partir de celui-ci est élagué.

En effet, si un ensemble de sommets n'est pas un générateur minimal alors aucun de ses sur-ensembles ne peut être aussi un générateur minimal. Comme les auteurs ont déjà montré qu'une traverse minimale est nécessairement un générateur minimal, il est inutile de considérer ces sur-ensembles.
La deuxième stratégie d'élagage consiste à éliminer les sur-ensembles d'un ensemble de sommets qui est une traverse minimale, puisqu'ils ne vérifient pas la condition de minimalité au sens de l'inclusion.

\begin{algorithm}[!h]
{
\LinesNumbered
\Entree{$H$ = ($\mathcal{X}$, $\xi$): Hypergraphe}
\Sortie{$\mathcal{M}_H$ : ensemble des traverses minimales de $H$}

\Deb
{
$\mathcal{M}_H$ = \{\{$x$\} $\mid$ $x$ $\in$ $\mathcal{X}$ et $\mid$$gH$(\{$x$\})$\mid$ = 0\};\\
$Gen_1$ = \{\{$x$\} $\mid$ $x$ $\in$ $\mathcal{X}$ et $0$ $<$ $\mid$$gH$(\{$x$\})$\mid$ $<$ $\mid$$\xi$$\mid$ \};\\
k = 1;\\
\Tq{$Gen_k$ $\neq$ $\emptyset$}{
\PourCh{($X \cup$ \{$x_1$\}, $X \cup$ \{$x_2$\}) $\in Gen_k \times Gen_k$}
{
$Z$ = $X$ $\cup$ $x_1$ $\cup$ $x_2$;\\
$gH$($Z$) = $gH$($X \cup$ \{$x_1$\}) $\cap$ $gH$($X \cup$ \{$x_2$\});\\
i = 1;\\
\Tq{i $\leq$ k+1 et $Z \backslash$\{$x_i$\} $\in Gen_k$ et $\mid$ $gH$($Z$) $\mid$ $<$ $\mid$ $gH$($Z \backslash$ \{$x_i$\}) $\mid$}{i = i+1;}
\Si{i = k + 2}{\eSi{$\mid$ $gH$($Z$) $\mid$ = 0}{$\mathcal{M}_H$ = $\mathcal{M}_H$ $\cup$ $Z$;}{$Gen_{k+1}$ = $Gen_{k+1} \cup$ \{$Z$\};}}
}
k = k + 1;\\
}
\Retour{$\mathcal{M}_H$}
}
}

  \caption{L'algorithme \textsc{MtMiner} \cite{thCline}}
  \label{mtminer}
\end{algorithm}

L'algorithme \textsc{MtMiner} effectue un parcours en largeur en démarrant le balayage de l'espace de recherche par les sommets, avant de générer les ensembles plus grands, suivant une approche inspirée de \textsc{Apriori} \cite{Agrawal12008994}. La stratégie en largeur permet de garantir la minimalité des candidats, dans la mesure où chaque candidat n'est gardé dans un niveau \emph{i} qu'après avoir calculé et testé l'extension de ses sous-ensembles directs, qui se trouvent dans le niveau \emph{i-1}.\\
L'ensemble $\mathcal{M}_H$ est initialisé avec les sommets d'extension vide, i.e, qui appartiennent à toutes les hyperarêtes de l'hypergraphe d'entrée ($\mid$$gH$(\{$x$\})$\mid$ = 0). Ces sommets représentent donc des traverses minimales du niveau $1$. A chaque niveau \emph{i}, \textsc{MtMiner} génére des candidats $Z$, à partir des éléments calculés au niveau \emph{i-1}. Si le candidat $Z$ vérifie la propriété d'anti-monotonie et si $\mid$$gH$(\{$Z$\})$\mid$ = 0, il est alors ajouté à $\mathcal{M}_H$, sinon il sera reversé dans $Gen_{i+1}$ et servira comme générateur pour le niveau \emph{i+1}. Si $Z$ ne vérifie pas la propriété d'anti-monotonie, il n'est donc pas un générateur minimal et il est tout simplement élagué de l'espace de recherche.

D'après Hébert \emph{et al.}, la complexité de l'algorithme \textsc{MtMiner} dépend de $\tau$($H$) et $\mid$ $\mathcal{M}_H$ $\mid$ \cite{thCline}. Pour chaque traverse minimale $T$, l'algorithme considère au plus $2^{\mid T \mid}$ ensembles de sommets et effectue, par conséquent, un nombre d'opérations inférieur à : $\sum_{T \in  \mathcal{M}_H} (2^{\mid T \mid})$.

De ce fait, pour un hypergraphe $H$, \textsc{MtMiner} calcule l'ensemble des traverses minimales $\mathcal{M}_H$ en $O(2^{\tau(H)} \times \mid \mathcal{M}_H \mid)$. Cependant, pour Hagen, la complexité réelle de l'algorithme \textsc{MtMiner} est de $O$($N^{\Omega(log(log N))}$), telle que \textit{N} est la taille de l'entrée et de la sortie de l'algorithme. Hagen présente le calcul détaillé de cette complexité dans \cite{thesehagen}. D'un point de vue performances, \textsc{MtMiner} présente des temps de traitements intéressants, notamment pour des hypergraphes denses et ayant un nombre de transversalité assez bas. Cependant, et comme souligné par les auteurs, \textsc{MtMiner} est assez gourmand en consommation mémoire \cite{bass2008}.

\section{Algorithmes de type \textsc{Shd} \cite{uno1}} \label{uno}
Murakami et Uno proposent les algorithmes de type \textsc{Shd}, \textsc{Mmcs} et \textsc{rs}, qui visent à réduire l'espace de recherche \cite{uno1}. En ce sens, ces algorithmes sont destinés à traiter des hypergraphes de grande taille constitués par un très grand nombre d'hyperarêtes.

Les algorithmes de type \textsc{Shd} adoptent une stratégie de parcours en profondeur de l'espace de recherche qui, dans le cas de \textsc{rs}, est équivalente à celle de l'algorithme de Kavvadias et Stavropoulos. La principale différence entre ce dernier et \textsc{rs} repose sur l'élimination des itérations redondantes où aucun sommet n'est ajouté à un ensemble de sommets générés auparavant. De plus, Murakami et Uno introduisent deux nouveaux concepts, i.e, le test de la transversalité (\emph{uncov}) et les hyperarêtes critiques (\emph{crit}), et ce afin d'optimiser les tests sur la minimalité effectués sur l'ensemble des traverses générées.

\begin{algorithm}[!h]
{
\LinesNumbered
\textit{\textbf{Var. Globale}} : \emph{uncov} (initialisé à $\xi$), $Cand$ (initialisé à $\mathcal{X}$), \emph{crit}[$x$] initialisé à $\emptyset$ pour chaque $x$\\
\Entree{$H$ = $(\mathcal{X}$, $\xi)$: Hypergraphe, $X$ : ensemble de sommets}
\Sortie{$T$ tel que $T$ $\in$ $\mathcal{M}_H$}

\Deb
{
\Si{\emph{uncov} = $\emptyset$}{\Retour{$X$}}
Choisir une hyperarête $e$ à partir de \emph{uncov};\\
$C$ = $Cand$ $\cap$ $e$;\\
$Cand$ = $Cand$ $\backslash$ $C$;\\
\PourCh{$x$ $\in$ $C$}{\textsc{Update\_crit\_uncov}($x$, \emph{crit, uncov});\\
\Si{\emph{crit}(f, $X$ $\cup$ $x$) $\neq$ $\emptyset$ \textbf{pour chaque} f $\in$ $X$}{\textsc{Mmcs}($X$ $\cup$ $x$);\\ $Cand$ = $Cand$ $\cup$ $x$;}
Restaurer les valeurs de crit et \emph{uncov} d'avant la ligne $8$;
}

}
}

  \caption{L'algorithme \textsc{Mmcs} \cite{uno1}}
  \label{unoalgo}
\end{algorithm}

Étant donné $X$ un ensemble de sommets, éventuellement réduit à un singleton, \emph{uncov($X$)} désigne l'ensemble des hyperarêtes que n'intersectent pas $X$, i.e, \emph{uncov($X$)} = \{$e$ $\in$ $\xi$, $e$ $\cap$ $X$ = $\emptyset$\}. $X$ est une traverse si et seulement si \emph{uncov($X$)} = $\emptyset$.

Pour un sommet $x$ $\in$ $X$, une hyperarête $e$ $\in$ $\xi$ est dite \emph{critique} pour $x$ si $X$ $\cap$ $e$ = \{$x$\}. L'ensemble des hyperarêtes critiques pour $x$ est noté \emph{crit($x$, $X$)}, i.e., \emph{crit($x$, $X$)} = \{$e$ $\mid$ $e$ $\in$ $\xi$, $e$ $\cap$ $X$ = \{$x$\}\}. Ainsi, un ensemble de sommets $X \subseteq \mathcal{X}$ est une traverse si \emph{uncov($X$)} = $\emptyset$ et c'est une traverse minimale si, en plus, \emph{crit($x$, $X$)} = $\not\emptyset$ $\forall x \in X$.

Par ailleurs, si $X$ est une traverse, alors si le sommet $x$ n'a aucune hyperarête critique, chaque $e$ $\in$ $\xi$ renferme un sommet de $X$, autre que $x$, et $X$ $\backslash$ $x$ est alors une traverse. Ceci est résumé par la propriété \ref{propUno} proposée par les auteurs.

\begin{property}\label{propUno}
$X$ est une traverse minimale si et seulement si \emph{uncov($X$)} = $\emptyset$ et \emph{crit($x$, $X$)} $\neq$ $\emptyset$, $\forall$ $x$ $\in$ $X$.
\end{property}

Les auteurs proposent aussi divers lemmes, dans \cite{uno1}, pour optimiser le calcul de la fonction \emph{crit}, qui est la clé de leur approche.

Les algorithmes de type \textsc{Shd} se basent donc sur la même approche et l'Algorithme \ref{unoalgo} décrit le pseudo-code de l'algorithme \textsc{Mmcs}. Cet algorithme est récursif et fournit en sortie des traverses minimales en série. Pour tester un ensemble de sommets $X$, les algorithmes cherchent, de façon itérative, les sous-ensembles de $X$ et effectuent un appel récursif pour chacun tout en mettant à jour les ensembles \emph{crit} et \emph{uncov}. En opérant de cette manière, Murakami et Uno permettent à leur algorithme de balayer l'espace de recherche en profondeur en cherchant seulement les sous-ensembles du candidat courant. La méthode et les étapes pour la recherche des sous-ensembles d'un candidat sont détaillées dans \cite{uno1}. L'étude expérimentale effectuée par les auteurs a montré que les algorithmes de type \textsc{Shd} (et notamment \textsc{mmcs}) présentaient des performances très intéressantes et s'imposaient comme les algorithmes les plus performants dans la littérature.

\section{Algorithme de Toda \cite{toda13}}
L'algorithme de Toda est le plus récent dans la littérature \cite{toda13}. Cet algorithme fait appel à des structures de données compressées qui permettent d'exploiter les capacités des diagrammes de décision binaire (\textsc{bdd}) et une des améliorations de ces dernières, i.e., les zéro diagrammes supprimés de décision (\textsc{zdd}). Les diagrammes de décision binaires permettent de représenter des fonctions booléennes sous la forme de graphes orientés. Leurs principes et mécanismes de fonctionnement sont décrits dans les travaux de \cite{Akers78} et \cite{BraceRB90}.

Toda se base sur les travaux de Donald Knuth sur les \textsc{zdd} et tente de les additionner aux \textsc{bdd} pour optimiser son algorithme. L'intérêt des \textsc{bdd} dans l'algorithme de Toda se trouve dans la représentation des résultats intermédiaires. Comme le montre le pseudo-code de l'Algorithme \ref{todaalgo}, Toda génère d'abord les traverses avant de tester leur minimalité. Les traverses candidats sont compressées et stockées dans un \textsc{bdd} avant que le \textsc{zdd} généré ne fournisse les traverses minimales souhaitées. Ainsi, dans l'Algorithme \ref{todaalgo}, $\mathcal{S}(p)$ dénote la famille des ensembles d'un \textsc{bdd} (ou un \textsc{zdd}).

\begin{algorithm}[!h]
{
\LinesNumbered
\Entree{$H$ = $(\mathcal{X}$, $\xi)$: Hypergraphe}
\Sortie{$\mathcal{M}_H$}

\Deb
{
$p$ = Compresser $\xi$ en un \textsc{zdd};\\
Calculer le \textsc{bdd} $q$ pour toutes les traverses de $\mathcal{S}(p)$;\\
Calculer le \textsc{zdd} $r$ pour tous les ensembles minimaux dans $\mathcal{S}(q)$;\\
Décompresser $r$ en un ensemble $\xi^{\star}$;\\
\Retour{$\xi^{\star}$}
}
}

  \caption{L'algorithme de Toda \cite{toda13}}
  \label{todaalgo}
\end{algorithm}

Les expérimentations effectuées par Toda visent à comparer son algorithme à celui de Murakami et Uno, présenté dans la section \ref{uno}. L'etude expérimentale a montré que l'algorithme de Toda est compétitif, y compris sur les bases éparses. Ceci peut être expliqué par les capacités qu'offrent les \textsc{zdd} sur ce type de bases.

\section{Les traverses minimales approchées}
Nous avons présenté, dans le chapitre précédent, le problème de l'extraction des traverses minimales d'un hypergraphe comme étant un problème NP-difficile. En ce sens, à côté des algorithmes présentés plus haut, d'autres travaux se sont intéressés à la recherche des traverses minimales approchées dans le but de contourner la difficulté du problème \cite{AbreuG09} \cite{Durand13}. Ces travaux, assez rares toutefois, s'intéressent à une sous-classe des traverses minimales dont les éléments n'intersectent pas toutes les hyperarêtes de l'hypergraphe d'entrée.
Ainsi, dans \cite{AbreuG09}, les auteurs se sont basés sur une approche évolutionnaire où la transversalité et la minimalité sont transcrites dans une fonction objective. D'autres approches introduisent un certain nombre d'exceptions liées à la transversalité pour générer les traverses minimales approchées.
Récemment, Durand \textit{et al.} ont présenté dans \cite{Durand13} l'algorithme $\delta$-\textsc{MTminer} qui permet de calculer les transverses minimales approchées. L'algorithme prend en compte un nouveau paramètre, $\delta$, qui correspond au nombre des hyperarêtes qu'une traverse minimale approchée pourrait ne pas intersecter.

\section{Discussion}
A la lumière de notre description des principaux algorithmes d'extraction des traverses minimales d'un hypergraphe et avec pour objectif de situer nos contributions, présentées dans les chapitres suivants, par rapport à ces travaux, nous avons synthétisé les caractéristiques de ces algorithmes dans le tableau \ref{synthese}. Les critères que nous avons choisis pour distinguer les différentes approches sont le principe, sur lequel se base chaque algorithme, la stratégie d'exploration et les techniques d'élagages.
\begin{table}[htbp]
\small{
\begin{tabular}{|r|r|r|r|}
\hline
   \textbf{Algorithme}  &  \textbf{ Algos sous-jacents} &   \textbf{Stratégie} &   \textbf{Techniques d'élagages} \\
        &   &   \textbf{d'exploration} &   \\
\hline
  \textsc{Berge} & - Processus &  En largeur & Aucune \\
     & incrémental &   &  \\
\hline
  Dong et Li & - Algorithme \ref{algoBerge} de Berge  & En largeur & Traverses garanties $T_g$ \\
     & - Itemsets emergeants  &  &  Couvertures d'hyperarêtes $e^{cov}$\\
\hline
  Kavvadias et & - Algorithme \ref{algoBerge} de Berge & En profondeur & \textit{Sommets généralisés}\\
  Stavropoulos   &  &  & \textit{Sommets appropriés}\\
\hline
  Bailey \textit{et al.} & - Algorithme \ref{algoBerge} de Berge et & En largeur & - \\
       & Algorithme \ref{Dong} de Dong et Li & & \\
     & - Partitionnement des & & \\
    & hyperarêtes &  & \\
  \hline
  FK & - Dualisation des formules  & - & Dualité Mutuelle\\
   &  booléennes monotones &  & \\
  \hline
  \textsc{MtMiner} & - Extraction de motifs & En largeur & Anti-monotonie de \\
  &  &  & la minimalité \\
  \hline
  \textsc{Mmcs} & - Itemsets fermés & En profondeur & \textit{uncov} et \textit{crit}\\
  \hline
  \textsc{Toda} & - Structures de & En profondeur & Caractéristiques des \\
     & données compressées & & \textsc{bdd} et \textsc{zdd} \\
  \hline
\end{tabular}}
  \caption{Caractéristiques des algorithmes de l'état de l'art}\label{synthese}
\end{table}

Les caractéristiques des principaux algorithmes d'extraction des traverses minimales ont été établies à partir des différentes sections présentées dans ce chapitre. La première constatation qui se dégage de ce tableau est qu'aucune approche n'a mis à profit la notion de nombre de transversalité, introduite par la Définition \ref{TMberge}. Cette notion qui donne une indication claire sur le nombre minimum de sommets formant une traverse minimale peut être intéressante, notamment en adoptant une stratégie en largeur pour cibler directement le niveau qui contient ces plus petites traverses minimales. Notre première contribution, qui consiste à détecter les multi-membres d'un réseau social et qui correspondent à des plus petites traverses minimales d'un hypergraphe représentant les différentes communautés d'un réseau social, se base sur cette notion de nombre de transversalité pour optimiser l'extraction des plus petites traverses minimales à travers un algorithme appelé \textsc{om2d} que nous introduisons dans le chapitre suivant.

Une deuxième constatation concerne les éléments générés en sortie par les différents algorithmes. Tous ces derniers calculent toutes les traverses minimales et leurs cardinalités sont généralement très importantes. Ce nombre de traverses minimales pouvant s'avérer exponentiel en la taille de l'hypergraphe, nous nous sommes intéressés à chercher une forme d'irrédondance dans l'ensemble des traverses minimales. Le fait de représenter cet ensemble de manière concise et exacte améliore le temps de traitement nécessaire à l'extraction de toutes les traverses minimales. En outre, et partant du fait que les traverses minimales apportent des solutions dans de nombreuses applications, comme présenté dans la section \ref{DomAppli} du chapitre \ref{chap1} (page \pageref{DomAppli}), cette représentation concise a des répercussions directes sur l'optimisation de bien d'autres problématiques. Notre deuxième contribution, présentée dans le chapitre $4$, met en avant cette notion d'irredondance qui se cache dans les traverses minimales et l'illustre en présentant son impact sur le problème du calcul de la couverture minimale des dépendances fonctionnelles en bases de données.

Enfin, notre troisième contribution consiste en une optimisation de l'extraction des traverses minimales en adoptant la stratégie \textit{diviser pour régner}. Cette stratégie a été utilisée par l'algorithme \ref{bailey} de Bailey \emph{et al.} sauf que les auteurs se sont focalisés sur le partitionnement des hyperarêtes quand celles-ci sont composées d'un nombre important de sommets. Notre idée se base, plutôt, sur le partitionnement de l'hypergraphe d'entrée en $k$ hypergraphes partiels tel que $k$ est égal au nombre de transversalité. De cette manière, nous éliminons le test coûteux de la minimalité sur les traverses formées par $k$ sommets. Cette approche s'est avérée fructueuse mais, seulement, sur un certain type d'hypergraphes. Une étude détaillée est présentée dans le chapitre $5$.

\section{Conclusion}
Dans ce chapitre, nous avons présenté les principaux algorithmes de calcul des traverses minimales, proposés dans la littérature. Que ce soit en adoptant une stratégie en largeur d'abord ou en profondeur d'abord, les différentes approches ont tenté d'innover avec pour but commun d'optimiser l'extraction des traverses minimales. En profitant de la notion de nombre de transversalité, inutilisée jusque-là, nous nous intéressons, dans le chapitre suivant, à une classe particulière des traverses minimales et à son application dans les systèmes communautaires.

\chapter{Identification des multi-membres dans un réseau social}\label{chap3}
\section{Introduction}
Avec l'expansion des systèmes communautaires du Web 2.0, beaucoup de travaux se sont intéressés à identifier les membres clés dans les réseaux sociaux, qualifiés, selon les auteurs, d'influenceurs, de médiateurs, d'ambassadeurs ou encore d'experts. Ce problème a été notamment considéré comme un problème de maximisation. Dans ce chapitre, nous présentons un type particulier d'acteurs que nous appelons multi-membres, en raison de leur appartenance à plusieurs communautés. Nous introduisons alors un cadre méthodologique pour identifier ce type d'acteurs dans un hypergraphe, dans lequel les sommets sont les acteurs et les hyperarêtes représentent les communautés. Nous démontrons que détecter les multi-membres pourrait être ramené au problème d'extraction des traverses minimales à partir d'un hypergraphe et nous présenterons deux algorithmes d'extraction des multi-membres qui conjuguent des concepts de la fouille de données et de la théorie des hypergraphes. Au cours de l'étude expérimentale, nous étudierons notamment la nature des acteurs qui constituent une traverse minimale multi-membres et leurs rôles au sein du réseau.

\section{Problématique}
C'est en s'appuyant sur des représentations et des concepts issus de la théorie des graphes que les réseaux sociaux ont été étudiés en sciences sociales dès les années soixante \cite{moreno34} \cite{Cartwright77}.
Parmi les questions essentielles que l'analyse de réseau s'efforce de traiter figure l'identification d'individus occupant un rôle déterminant dans le réseau. Ainsi, plusieurs indicateurs tels que la \emph{centralité} ou le \emph{prestige} ont été définis pour caractériser la position occupée par un acteur \cite{Freeman} \cite{Scott2000} \cite{WassermanFaust1994} \cite{Borgatti1992}. Néanmoins, la communauté scientifique informatique a été confrontée au problème du passage à l'échelle des algorithmes classiques de détection de tels individus.
Avec l'émergence du Web 2.0 et l'explosion des réseaux sociaux sur Internet \cite{hamdi2013trust}, des travaux plus récents se sont attachés à repérer des acteurs qui occupent une place particulière dans le réseau et qui selon appelés, selon les auteurs, \emph{influenceurs}, \emph{médiateurs}, \emph{ambassadeurs} ou encore les \emph{experts} \cite{Domingos2006}, \cite{Scripps2007}, \cite{Scripps2007a}, \cite{NitinAgarwal2008}, \cite{opsahl2010}.

L'identification de tels acteurs a eu, en effet, de nombreuses applications, e.g., dans les domaines de l'épidémiologie, du marketing, ou encore de la diffusion d'innovation.

En particulier, plusieurs algorithmes ont été présentés récemment \cite{Leskoveckdd2007}, \cite{Chenkdd2009}, \cite{YuWang2010}, \cite{Chenkdd2010}, \cite{Chenicdm2010}, \cite{KimuraS06}, \cite{Goyalwww11} pour résoudre le problème de recherche d'influenceurs, redéfini comme un problème de maximisation \cite{DomingosR01}, \cite{RichardsonD02}, \cite{Kempe1}.\\


  Parmi les modèles de diffusion dans un réseau, répertoriés dans la littérature, on peut distinguer d'une part les modèles linéaires à seuil inspirés des travaux de Granovetter et Schelling \cite{Granovetter78}, \cite{Schelling78} et d'autre part les modèles à  cascade indépendantes \cite{Goldenberg01}. Dans tous ces modèles, on considère qu'à un instant donné, chaque membre du réseau est soit actif soit inactif.  On cherche, par un processus itératif, à déterminer les acteurs devenus actifs à partir d'un sous-ensemble d'acteurs initialement actifs. On suppose bien sûr qu'un acteur peut être influencé par ses voisins suivant un certain seuil ou une certaine probabilité.
La mise en oeuvre de ces modèles requiert donc l'estimation de ces probabilités d'influence ou de ces seuils. 
Cependant, l'estimation des paramètres n'est pas le seul inconvénient de ces modèles.\\
En effet, la recherche des influenceurs dans un réseau peut être énoncée plus formellement comme un problème d'optimisation discrète, connu dans la littérature sous le nom de "\emph{influence maximization}" ou "\emph{spread maximisation}". Étant donné $A$, un ensemble d'acteurs du réseau et une mesure d'influence associée à cet ensemble, définie comme le nombre des acteurs devenus actifs à partir de $A$, le problème revient à déterminer, pour un paramètre k donné, les k acteurs du réseau qui maximise la fonction d'influence.
Or, Kempe \textit{et al.} ont démontré que ce problème était NP-complet pour les deux familles de modèles citées précédemment \cite{Kempe1}.

En s'appuyant sur la théorie des fonctions submodulaires, Kempe a aussi défini un cadre d'analyse généralisant les modèles à cascade et les modèles à seuil, et a montré qu'il est possible de déterminer une solution qui approche la solution optimale à un facteur près, à l'aide d'algorithmes gloutons, tels que \textit{Greedy Algorithm} \cite{Kempe1}. Ceci a conduit au développement d'heuristiques permettant de déterminer approximativement les influenceurs dans un réseau. Ainsi, en suivant ce cadre d'analyse, Leskovec \textit{et al.} ont développé l'algorithme "Cost-Effective Lazy Forward" (CELF), qui a donné lieu ensuite à plusieurs extensions, telles que NewGreedy et MixedGreedy introduit par Chen \textit{et al.} ou, plus récemment, CELF++  par Goyal \textit{et al.} \cite{Leskoveckdd2007}, \cite{Chenkdd2009}, \cite{Goyalwww11}.\\
L'algorithme \textit{Greedy} a fait aussi l'objet d'autres améliorations, exploitant des propriétés spécifiques du modèle à cascade \cite{KimuraS06} \cite{Chenkdd2010} ou du modèle à seuil \cite{Chenicdm2010}.
D'autres solutions ont aussi été proposées pour résoudre le problème de maximisation de l'influence, comme par exemple le  modèle de vote de Even-Dar \textit{et al.}, qui exploite d'ailleurs les mêmes hypothèses que les modèles à seuil \cite{Even-Dar2007}.
Cependant, tous les travaux cités précédemment considèrent la recherche d'influenceurs comme un problème d'optimisation sans tenir compte explicitement des communautés présentes dans le réseau.
Or, comme le souligne Scripps \textit{et. al.}, il peut être utile, pour identifier des influenceurs, de mieux connaître les positions occupées par les acteurs au sein des communautés présentes dans le réseau \cite{Scripps2007a}, \cite{Scripps2007}.
C'est d'ailleurs le principe de l'algorithme \textit{Community-based Greedy}, qui consiste justement à détecter des communautés en tenant compte du processus de diffusion au sein du réseau puis à identifier les influenceurs au sein des communautés \cite{YuWang2010}.
Cependant, comme les algorithmes précédemment cités, \textit{Community-based Greedy} suppose que le réseau est décrit par un graphe simple de sorte que les relations entre les acteurs pris deux à deux sont connues. Ainsi, ces algorithmes exploitent la matrice d'adjacence associée au graphe décrivant le réseau. Dans de nombreuses applications, on ne dispose pas forcément de cette information. Par contre, on sait à quelle(s) communauté(s) appartient un acteur. Ainsi par exemple, on sait quels sont les chercheurs qui ont participé à la conférence \emph{KDD} et ceux qui ont assisté à \emph{VLDB} sans forcément connaître les liens directs existants entre ces chercheurs. De même, dans le domaine du marketing, on peut savoir quels sont les clients qui ont acheté des articles d'une gamme de produits sans savoir s'ils sont en relation. Dans ces deux cas, il peut être intéressant d'identifier des acteurs, en nombre le plus petit possible, susceptibles de diffuser des idées ou recommandations d'un groupe à un autre. C'est ce problème que nous nous proposons de résoudre. Nous émettons l'hypothèse que la propagation repose sur des acteurs qui sont susceptibles d'assurer la transmission entre les groupes d'individus. En ce sens, les multi-membres que nous recherchons sont, pour partie,  des ambassadeurs tels que Scripss \textit{et al} les définissent \cite{Scripps2007a}. Il s'agirait donc de déterminer le plus petit ensemble de membres du réseau susceptibles de couvrir toutes les communautés.

 L'objectif est donc de déterminer le plus petit ensemble d'acteurs, appelés \emph{multi-membres}, qui sont susceptibles de représenter au mieux possible les différentes communautés du réseau en analysant le réseau dans ce contexte d'information incomplète où nous ne disposons pas de la matrice d'adjacence associée au graphe représentant le réseau mais où, en revanche, les communautés sont données.

Pour ce faire, nous proposons de représenter le système communautaire sous la forme d'un hypergraphe dans lequel les sommets représentent les acteurs et les hyperarêtes représentent les communautés. Dans cet hypergraphe, les multi-membres pourront être déterminés à partir des traverses minimales de l'hypergraphe, elles-mêmes définies comme un ensemble de sommets, minimal au sens de l'inclusion, qui intersecte toutes les hyperarêtes \cite{berge1989}.

\section{Definition d'une traverse minimale multi-membres}
\label{newdef}
Un réseau social peut être défini comme un ensemble d'entités interconnectés les unes aux autres \cite{WassermanFaust1994}. Ces entités sont généralement des individus ou des organisations. Les relations matérialisent les interactions entre les entités. \\
Dans le contexte d'un hypergraphe où nous disposons seulement des communautés d'un réseau modélisées par les hyperarêtes, nous considérons que les multi-membres correspondent au plus petit ensemble de sommets tel qu'au moins un élément appartient à chaque communauté et, si possible, avec plusieurs éléments appartenant à des communautés de large taille. En ce sens, la définition des multi-membres peut se baser sur celle d'une traverse minimale définie comme étant un ensemble de sommets ayant une intersection, non vide, avec chaque hyperarête.

\begin{example}\label{exp_hyp}
Dans la suite de ce chapitre, nous utiliserons à titre illustratif l'hypergraphe $H = (\mathcal{X}, \xi)$ de la Figure \ref{hyp11} (gauche) tel que $\mathcal{X} = \{1,2,3,4,5,6,7,8\}$ et $\xi$ = \{\{$1,2$\}, \{$2,3,7$\}, \{$3,4,5$\}, \{$4, 6$\}, \{$6,7,8$\}, \{$7$\}\}. L'ensemble de ses traverses minimales $\mathcal{M}_H$ est \{\{$1,4,7$\}, \{$2,4,7$\}, \{$1,3,6,7$\}, \{$1,5,6,7$\}, \{$2,3,6,7$\}, \{$2,5,6,7$\}\}. La Table de la figure \ref{hyp11} (droite) représente la matrice d'incidence associée à l'hypergraphe $H$. Cet exemple d'hypergraphe $H$ sera repris tout au long de ce chapitre pour illustrer notre approche et extraire les multi-membres de $H$.
\end{example}

\begin{figure}
\begin{minipage}{.4\textwidth}\centering
\includegraphics[scale=0.7]{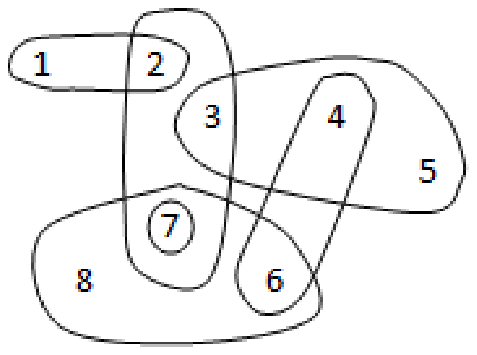}
\end{minipage}\hfill
\begin{minipage}{.6\textwidth}\centering
\small{
\begin{tabular}{|r|r|r|r|r|r|r|r|r|}
\hline
    {\bf } &   {\bf $\textbf{1}$} &   {\bf $\textbf{2}$} &   {\bf $\textbf{3}$} &   {\bf $\textbf{4}$} &   {\bf $\textbf{5}$} &   {\bf $\textbf{6}$} &   {\bf $\textbf{7}$} &   {\bf $\textbf{8}$}\\
\hline
  {\bf $\textbf{\{1, 2\}}$} &     $1$ &       $1$     &       0     &  0  &0 &0  & 0 &  0      \\
\hline
  {\bf $\textbf{\{2, 3, 7\}}$} &     0 &       $1$     &       $1$     &  0 &0 &0  & $1$ &  0 \\
\hline
  {\bf $\textbf{\{3, 4, 5\}}$} &     0 &       0     &       $1$     &  $1$  &$1$ & 0  & 0 &  0   \\
\hline
  {\bf $\textbf{\{4, 6\}}$} &     0 &       0     &       0     &  $1$  & 0 & $1$  & 0 &  0      \\
  \hline
    {\bf $\textbf{\{6, 7, 8\}}$} &     0 &       0     &       0     &  0  &0 & $1$  & $1$ &  $1$   \\
  \hline
    {\bf $\textbf{\{7\}}$} &     0 &       0     &       0     &  0  &0 &0  & $1$ &  0    \\
  \hline
\end{tabular}}
\end{minipage}
\caption{Un exemple d'hypergraphe $H = (\mathcal{X}, \xi)$ et la matrice d'incidence $IM_H$ correspondante}\label{hyp11}
\end{figure}

Nous pouvons redéfinir une traverse minimale à partir de la notion d'ensemble de sommets essentiels.

\begin{definition}\label{definitionsupportmotif}
\textsc{Support d'un ensemble de sommets} \\
Soit l'hypergraphe $H = (\mathcal{X}$, $\xi)$ et $X$ un sous-ensemble de sommets de $\mathcal{X}$. Nous définissons $Supp\textsc{(}\textit{X}\textsc{)}$ comme le nombre d'hyperarêtes de $H$, renfermant au moins un élément de $X$ :
$Supp\textsc{(}\textit{X}\textsc{)}$ =  $|\{e \in \xi | \exists x \in X  \wedge  \mathcal{R}(e, x) = 1\}|$, $\mathcal{R}$ $\subseteq$ $ \xi \times \mathcal{X}$ étant la relation binaire entre les hyperarêtes et les sommets de la matrice d'incidence correspondante à $H$.

\end{definition}

Ainsi, l'ensemble $X$ peut être vu comme une disjonction de sommets $(x_1 \vee x_2  \vee \ldots \vee x_n)$  tel que la présence d'un seul sommet de $X$ suffit à affirmer que $X$ satisfait une hyperarête donnée, indépendamment des autres sommets.

\begin{definition}\label{defessentiel}
\textsc{Ensemble essentiel de sommets}~(\cite{casali2005})
  Soit l'hypergraphe $H = (\mathcal{X}$, $\xi)$ et $X \subseteq \mathcal{X}$. $X$ représente un ensemble essentiel de sommets si et seulement si : $ Supp(X) \mbox{ }> \mbox{ }\max\{Supp( X\backslash x)\mid
\forall x\in X\}$.
\end{definition}

Il est important de souligner que les ensembles essentiels, extraits à partir d'une matrice d'incidence, vérifient la propriété d'idéal d'ordre des itemsets essentiels, \emph{i.e}, si $X$ est un ensemble essentiel, alors $\forall Y \subset X$, $Y$ est aussi un ensemble essentiel.
De plus, la notion de traverse peut être redéfinie par le biais du support d'un ensemble de sommets et de la notion d'ensemble essentiel, selon la proposition \ref{Traverse2}.

\begin{proposition}\label{Traverse2}
\textsc{Traverse minimale}\\
Un sous-ensemble de sommets $X $ $\subseteq$ $\mathcal{X}$ est une traverse minimale de l'hypergraphe $H$, si $X$ est essentiel et si son support est égal au nombre des hyperarêtes de $H$, autrement dit, $X$ est un ensemble essentiel tel que $Supp(X)$=$|\xi|$.
\end{proposition}

\begin{proof}
Soit $X$ un ensemble essentiel de sommets tel que $Supp(X)$=$|\xi|$. Par conséquent, $X \bigcap e_i \neq  \emptyset$~ $\forall e_i \in \xi, i =1$, \ldots, $m$. Donc, d'après la définition \ref{TMberge}, $X$ est une traverse. La minimalité de $X$ tient à son "essentialité". En effet, puisque $X$ est essentiel, alors son support est strictement supérieur à celui de ses sous-ensembles directs. Par conséquent, $\nexists X_1 \subset X$ s.t. $Supp(X_1)$=$|\xi|$. $X$ est donc une traverse minimale.
\end{proof}

\begin{example}
L'ensemble des traverses minimales $\mathcal{M}_H$, calculées à partir de l'hypergraphe de l'Exemple \ref{exp_hyp}, est \{\{$1,4,7$\}, \{$2,4,7$\}, \{$1,3,6,7$\}, \{$1,5,6,7$\}, \{$2,3,6,7$\}, \{$2,5,6,7$\}\}.
\end{example}

En se basant sur les définitions présentées ci-dessus, nous pouvons donner une définition plus formelle des traverses minimales multi-membres (\textsc{tmm}).

\begin{definition}\label{Deff}
\textsc{Traverse minimale multi-membres}\\
  Soit $H = (\mathcal{X}$, $\xi)$, un hypergraphe et $X \subset \mathcal{X}$.  $X$ est appelé \textit{Traverse minimale multi-membre}, noté \textsc{Tmm},  si $X$ vérifie les trois conditions suivantes :
 \begin{enumerate}
   \item (\textbf{Condition nécessaire}): $X$ est une traverse minimale
:  $X \in$ $\mathcal{M}_H$.
  \item (\textbf{Condition de composition}):  $X$ est minimale dans $\mathcal{M}_H$ dans le sens de la cardinalité
 : $|$X$|$ = $\tau$(H) where $\tau$(H) = Min $\{$$|$T$|$, $\forall$T$\in$$\mathcal{M}_H$$\}$. $\tau$($H$) est le nombre de transversalité de $H$.
 \item  (\textbf{Condition de recouvrement maximum}):\\
$\sum \limits_{e_i \in \xi / e_i \bigcap X \neq \emptyset} |e_i|$ = Max $\{$ $\sum \limits_{e_i \in xi / e_i \bigcap T \neq \emptyset} |e_i|$, $\forall$ T $\in$ $\mathcal{M}_H$ tel que $|$T$|$ = $\tau$(H) $\}$.
 \end{enumerate}
\end{definition}

Ainsi, un ensemble de sommets est une \textsc{Tmm} s'il constitue une traverse minimale, si sa taille est la plus petite possible et s'il maximise la condition de recouvrement.
Spécifiquement, la première condition assure qu'il existe au moins un multi-membre dans chaque communauté. La seconde suppose que l'ensemble des multi-membres est le plus petit possible. Ainsi, l'objectif est de représenter toutes les communautés avec un nombre minimal de sommets.
La troisième condition, calculant le recouvrement maximum, prend en compte le fait que certains multi-membres peuvent appartenir à une ou plusieurs mêmes communautés. Dans ce cas, le but est de favoriser les éléments qui appartiennent aux communautés les plus grandes.

\begin{example}
Pour l'hypergraphe de l'Exemple \ref{exp_hyp}, nous avons une seule \textsc{Tmm}. C'est la traverse minimale \{$2,4,7$\}.
\end{example}

\section{Méthodologie et algorithmes d'extraction des multi-membres}\label{algo}

A présent, nous introduisons un premier algorithme d'extraction des multi-membres, baptisé \textsc{M2D}, qui balaye l'espace de recherche en largeur. \textsc{M2D} repose sur la propriété d'ordre idéal garantie par les ensembles de sommets essentiels pour l'élagage des candidats. En ce sens, cet algorithme agit d'une manière brute-force en générant les candidats nécessaires. Il s'arrête après avoir atteint le niveau $k$, i.e. le \emph{nombre de transversalité}, où a été détecté la première traverse minimale.

\subsection{Algorithme \textsc{M2D}}

L'algorithme \textsc{M2D}, dont le pseudo-code est décrit par l'algorithme \ref{algoglae}, prend en entrée un hypergraphe $H$ et fournit en sortie l'ensemble des \textsc{Tmm}s. L'algorithme effectue un parcours en largeur d'abord, i.e., il opère par niveau pour déterminer les ensembles de sommets essentiels. A chaque niveau $k$, un appel à la procédure \textsc{Apriori-Gen} \cite{Agrawal12008994} est effectué pour calculer les candidats de taille $k$, à partir des ensembles de sommets essentiels de taille $k-1$. En effet, \textsc{Apriori-Gen} génère un nouveau candidat $X''$ = $\{x_1, x_2, \ldots, x_{i-1}, x_i, x_{i+1}\}$ à partir de deux candidats $X'$ et $X$, tels que $X'$ = $\{x_1, x_2, \ldots, x_{i-1}, x_{i}\}$ et $X$ = $\{x_1, x_2, \ldots, x_{i-1}, x_{i+1}\}$. \textsc{M2D} calcule, ensuite, le support des k-candidats générés, à la ligne 9, et vérifie si leurs supports respectifs sont strictement supérieurs à ceux de leurs sous-ensembles directs (ligne 10).

Si parmi les ensembles de sommets essentiels, générés à un niveau $k$, il existe au moins un sommet, dont le support est égal au nombre d'hyperarêtes de l'hypergraphe (ligne 12), la boucle de la ligne 8 s'arrête et $\mathcal{TMM}$ renferme alors l'ensemble des traverses minimales qui sont minimales au sens de la cardinalité.

  \begin{algorithm}[!h]\label{algoglae}
 \LinesNumbered
\Entree{$H$ = ($\mathcal{X}$, $\xi$): Hypergraphe}
 \Sortie{$\mathcal{TMM}$}
 \Deb{
 $\mathcal{TMM}$ = $\emptyset$;\\
       $i$ := 1\;
    \PourCh{$x \in \mathcal{X}$}{\Si{$Supp\textsc{(}\textit{x}\textsc{)}=|\xi|$}
     {     $\mathcal{TMM} = \mathcal{TMM} \cup \{x\}$\;
     }
     }
     \eSi{$\mathcal{TMM}$ $\neq$ $\emptyset$}{Aller ligne $19$}{
     find = false \;
    \Tq{$L_{i}\neq\emptyset$ \textbf{ou} find = false}{
        $C_{i+1}$ := \textsc{Apriori-Gen}$(L_{i})$\;

        $L_{i+1}$ := $\{X\in C_{i+1}\mid\nexists$ $x\in X:Supp(X)=Supp(X\backslash x)\}$\;
               \PourCh{$X\in L_{i+1}$ }
          {
               \Si{$Supp(X)=|\xi|$}      {     $\mathcal{TMM} = \mathcal{TMM} \cup \{X\}$\;

               find = true;

   }
    }
        $i$ := $i$ + 1\;

    }
    }
    $\mathcal{TMM}$ = \textsc{Recouvrement}($\mathcal{TMM}$);\\
    \Retour{ $\mathcal{TMM}$ }

    }

     \caption{\textsc{M2D}}
\end{algorithm}

 Ces ensembles de sommets vérifient aussi bien la condition nécessaire que la condition de composition de la définition \ref{Deff}. Au final, parmi ces candidats, les \textsc{Tmm}s sont déterminés en se basant sur la fonction de calcul du recouvrement (ligne $16$). Cette fonction calcule le nombre de sommets couverts par chaque candidat (\textit{i.e.}, la somme des cardinalités des communautés auxquelles appartient ce candidat) et retourne ceux qui ont la valeur maximale. Ainsi, la troisième condition de la définition \ref{Deff} est aussi vérifiée.

\begin{table}[!h]
\begin{center}

\begin{tabular}{|c|c|}
\hline
        \textsc{Min. Tran.} &     Recouvrement \\
\hline
     \textbf{1 4 7} &          \texttt{\textbf{8}} \\
\hline
     \textbf{2 4 7}$\bigstar$ &        \texttt{ \textbf{10}} \\
\hline
   1 3 6 7 &            \\
\hline
   1 5 6 7 &            \\
\hline
  2 3 6 7  &            \\
\hline
   2 5 6 7 &            \\
\hline
\end{tabular}
\end{center}
   \caption{Les \textsc{Tmm}s extraits à partir de l'hypergraphe de la Figure \ref{hyp11}}\label{process}
  \end{table}

\begin{example}\label{exp_tmm}
Illustrons le déroulement de l'algorithme \textsc{M2D} sur l'hypergraphe de la Figure \ref{hyp11}. \textsc{M2D} balaye l'espace de recherche en opérant en largeur jusqu'à arriver au niveau 3. Toutes les traverses minimales, que renferme l'hypergraphe de la figure \ref{hyp11} (page \pageref{hyp11}), sont données par la première colonne du tableau \ref{process}. Seules les plus petites traverses minimales, au sens de la cardinalité, nous intéressent et c'est la raison pour laquelle l'algorithme s'arrête au troisième niveau. Ces traverses minimales sont marquées comme étant des \textsc{Tmm}s candidates : \emph{\{1 4 7\}} et \emph{\{2 4 7\}} dans le tableau \ref{process}. La fonction \textsc{Recouvrement} calcule, pour chaque candidat, la somme des tailles des communautés auxquelles il appartient et ne gardera que celui qui la maximise. Ainsi, le premier \textsc{Tmm} candidat est \emph{\{1, 4, 7\}}. Le sommet \emph{1} couvre le sommet \emph{2}, le sommet \emph{4} couvre les sommets \emph{3}, \emph{5} et \emph{6}. Enfin, le sommet \emph{7} couvre les sommets \emph{2}, \emph{3}, \emph{6} et \emph{8}. Le candidat \emph{\{1, 4, 7\}} couvre donc, deux fois, chacun des sommets \emph{2}, \emph{3} et \emph{6}, et une seule fois les sommets \emph{5} et \emph{8}. Ainsi, le recouvrement du candidat \textsc{Tmm} \emph{\{1, 4, 7\}} est égal à 8 sommets alors que, dans le même temps, le candidat \emph{\{2, 4, 7\}} couvre au total 10 sommets. Ce dernier candidat est alors la seule \textsc{Tmm} de l'hypergraphe d'entrée et est retourné par l'algorithme \textsc{M2D}.

\end{example}

L'algorithme \textsc{M2D} balaye l'espace de recherche du niveau 1 jusqu'au niveau $k$, i.e., le niveau renfermant les plus petites traverses minimales. Sachant que les \textsc{Tmm} appartiennent à un et un seul niveau, l'ensemble des candidats générés du niveau 1 jusqu'au niveau $k-1$ est inutile puisque ces derniers ne vérifient pas la condition nécessaire de la définition \ref{Deff}. Cette génération inutile des candidats handicape sérieusement l'efficacité du processus de recherche des \textsc{Tmm}s, spécialement dans les bases éparses où la taille des \textsc{Tmm}s est large (i.e., localisées dans un niveau élevé de l'espace de recherche). Idéalement, il serait plus bénéfique d'accéder directement à ce niveau $k$. Dans l'exemple ci-dessus, les candidats de taille 1 et 2 ne renferment pas des \textsc{Tmm}s puisque la taille de la plus petite traverse minimale est égale à 3. Ainsi, "sauter" les niveaux 1 et 2 présenterait une optimisation conséquente dans la mesure où le nouveau algorithme n'aura pas à générer et tester les candidats inutiles. Cet algorithme, appelé \textsc{O-M2D}, est une optimisation de l'algorithme \textsc{M2D} et détermine intelligemment le niveau adéquat $k$ pour identifier les \textsc{Tmm} à partir des k-candidats uniquement.

\subsection{Algorithme \textsc{O-M2D}}

Comme \textsc{M2D}, l'algorithme \textsc{O-M2D} prend en entrée un hypergraphe $H$ et donne en sortie l'ensemble des \textsc{Tmm}s. Le balayage s'effectue sur un seul niveau, i.e., le niveau qui renferme les \textsc{Tmm}s. \textsc{O-M2D} commence par invoquer la fonction \textsc{GetMinTransversality}, dont le pseudo-code est décrit par l'Algorithme \ref{proc}, pour localiser le niveau où la taille des candidats générés est égale à celle des \textsc{Tmm}s (ligne $2$). Ce niveau correspond au nombre de transversalité de l'hypergraphe d'entrée.

Ensuite, \textsc{O-M2D} génére, un à un, l'ensemble des $k$-candidats (ligne $3$). Pour chaque $k$-candidat, i.e., ensemble de sommets de taille $k$, l'algorithme calcule son support (ligne $4$). Si ce support est strictement supérieur au maximum des supports de ses sous-ensembles directs et qu'il est égal au nombre des hyperarêtes (ligne $5$), alors ce candidat est marqué comme étant une traverse minimale et donc une \textsc{Tmm} potentielle.

Quand toutes les traverses minimales ont été extraites, la fonction \textsc{Recouvrement} (ligne $7$) se charge d'identifier l'ensemble des \textsc{Tmm}s, comme expliqué plus haut.

\begin{algorithm}[htbp]\label{algoglae2}
 \LinesNumbered
\Entree{$H$ = ($\mathcal{X}$, $\xi$): Hypergraphe et $IM_H$ sa matrice d'incidence correspondante}
 \Sortie{$\mathcal{TMM}$}
 \Deb{
    $Level$ := \textsc{GetMinTransversality}($IM_H$);\\

               \PourCh{$X \subseteq \mathcal{X}$ tel que $\mid X \mid$ = $level$ }
          {
          \Si{$\nexists$ $x\in X:Supp(X)=Supp(X\backslash x)$} {
            \Si{$Supp(X)=|\xi|$} {
                $\mathcal{TMM} = \mathcal{TMM} \cup \{X\}$\; }
                }
            }
$\mathcal{TMM}$ = \textsc{Recouvrement}($\mathcal{TMM}$);\\
    \Retour{ $\mathcal{TMM}$ }
    }
     \caption{\textsc{O-M2D}}
\end{algorithm}

\begin{algorithm}[!h]
{
\LinesNumbered
\Entree{Matrice d'incidence $IM_{H}$ associée à $H$ = ($\mathcal{X}$, $\xi$)}
 \Sortie{$T$: Une plus petite traverse minimale de $H$ ; $k$ : Nombre de transversalité de $H$}
\Deb
{
$k$ = $\mid \xi \mid$;\\
$T$ = $\emptyset$;\\
\PourCh{$x$ $\in$ $\mathcal{X}$}{$i$ = 1;\\
$T_{tmp}$ = $\emptyset$;\\
$T_{tmp}$[$i$] = $x$;\\
$\xi'$ = $\xi$ $\backslash$ $\{e \in \xi \mid x \in e\}$;\\
($n$, $T_{tmp}$) = \textsc{hyp\_empty}($\xi'$, $\mid \xi \mid$, i, $T_{tmp}$);\\
\Si{$n$ $<$ $k$}{$k$ = $n$;\\ $T$ = $T_{tmp}$;}
}
\Retour{($k$, $T$)}
}
}
  \caption{\textsc{GetMinTransversality}}
  \label{proc}
\end{algorithm}

\begin{algorithm}[!h]
{
\LinesNumbered
\Entree{$\xi'$ : Ensemble d'hyperarêtes ; $min$, $i$ : entier ; $T_{tmp}$ : tableau de sommets}
 \Sortie{$min$ : Nombre minimum d'itérations pour obtenir un hypergraphe vide ; $T'$ : Ensemble de sommets de cardinalité égale à $min$}

\Deb
{
\eSi{$\xi'$ = $\emptyset$}{\Retour{(i, $T_{tmp}$)}}{$T$ = \{$x$ $\in \mathcal{X}$ \emph{tel que}
$ \mid \{e \in \xi' \mid  x \in e \} \mid$ = max
$\{\mid \{e \in \xi' \mid x_l \in e\} \mid, x_l \in \mathcal{X}\}$\};\\
$T'$ = $\emptyset$;\\
$i$ = $i$ + 1;\\
\PourCh{$x$ $\in$ $T$ }{$\xi''$ = $\xi'$ $\backslash$ $\{e \in \xi' \mid x \in e\}$;\\
$T_{tmp}$[$i$] = $x$;\\
($m$, $T_{tmp}$) = \textsc{hyp\_empty}($\xi''$, $min$, $i$, $T_{tmp}$);\\
\Si{$m$ $<$ $min$}{$min$ = $m$;\\$T'$ = $T_{tmp}$;}}
\Retour{($min$, $T'$)}
}
}
}
  \caption{\textsc{hyp\_empty}}
  \label{hypEmpty}
\end{algorithm}

La fonction \textsc{GetMinTransversality} recherche le nombre minimal de sommets pouvant constituer une traverse minimale, i.e. le nombre de transversalité de l'hypergraphe. Pour ce faire, la fonction  parcourt les sommets, un par un (ligne $3$). Pour chaque élément $x$ de $\mathcal{X}$, \textsc{GetMinTransversality} supprime de la matrice d'incidence $IM_{H}$ les hyperarêtes de $\xi$ qui contiennent $x$ (ligne $5$). Les hyperarêtes restantes sont stockés dans $\xi'$.
La fonction invoque ensuite \textsc{hyp\_empty}, dont le pseudo-code est donné par l'Algorithme \ref{hypEmpty}. \textsc{hyp\_empty} est une fonction récursive qui stocke dans $T$ les sommets ayant le plus grand support dans $\xi'$ (ligne $5$) et les traitera, un par un, en supprimant à chaque fois les hyperarêtes auxquelles appartient le sommet traité (ligne $8$). La condition d'arrêt de notre fonction récursive est l'absence d'hyperarêtes dans $\xi'$ (ligne $2$). La valeur stockée dans $m$ correspond au nombre d'appels à la fonction \textsc{hyp\_empty} nécessaires pour que $\xi'$ soit égal à l'ensemble vide. Pour chaque élément de $T$, la fonction vérifie si $m$ est la valeur trouvée jusque-là, parmi les éléments traités de $T$. Si tel est le cas, elle est stockée dans $min$ dont la valeur est retournée à la fin. Pour chaque sommet $x$ traité, l'ensemble $\xi'$ est réactualisé à toutes les hyperarêtes de l'hypergraphe auxquelles nous supprimons celles qui contiennent $x$.
Au final, \textsc{GetMinTransversality} retourne le nombre minimum d'itérations permettant de "vider" la matrice d'incidence. La valeur de $k$, retournée par la fonction, correspond ainsi au nombre de transversalité de l'hypergraphe d'entrée $H$.

\begin{conjecture}
La fonction \textsc{GetMinTransversality} permet d'obtenir une borne maximale du nombre de transversalité, noté $\tau$(H) dans la \textit{Definition} \ref{Deff}, d'un hypeergraphe. Dans le meilleur des cas, cette borne est exactement le nombre de transversalité.
 \end{conjecture}

\begin{example}
Reconsidérons l'exemple illustratif de la Figure \ref{hyp11}. En optimisant, comme nous l'avons expliqué précédemment, l'algorithme \textsc{M2D}, les sommets et 2-candidats ne sont pas générés et, donc, leurs supports ne sont pas calculés. \textsc{O-M2D} accède directement au niveau 3, générant tous les 3-candidats. Parmi ces 3-candidats, \textsc{O-M2D} détecte 24 ensembles de sommets essentiels mais seulement deux d'entre eux sont des traverses minimales : $\{1,2,7\}$ et $\{2,4,7\}$. En d'autres termes, seuls ces deux ensembles de sommets ont un support égal au nombre d'hyperarêtes. La fonction \textsc{recouvrement} permet de déterminer, au final, la ou les \textsc{Tmm}s. Le recouvrement du {Tmm} candidat {\{1, 4, 7\}} est égal à \emph{8} sommets alors que celui du candidat \emph{\{2, 4, 7\}} est égal à \emph{10} sommet. Donc, \textsc{O-M2D} sélectionne \emph{\{2, 4, 7\}} comme unique \textsc{Tmm}, en sortie.
\end{example}


\section{Etude de la complexité}
A partir d'un hypergraphe de $n$ sommets et $m$ hyperarêtes, nous avons :
\begin{enumerate}
  \item La fonction \textsc{GetMinTransversality} a une complexité exponentielle, au pire des cas, de $O$($m*n^{m+1}$), avec $n$ = $\mathcal{X}$ et $m$ = $\mid \xi \mid$, pour déterminer le nombre de transversalité.
  \item La cardinalité de l'ensemble des candidats de taille $k$ générés est égale à $C_n^k$.
  \begin{enumerate}
    \item \textsc{O-M2D} calcule, ensuite, le support de chaque sous-ensemble direct. Le support est obtenu en $m$ opérations.
    \item Le support de $X$ est calculé en $m$ opérations.
  \end{enumerate}
  \item Les tests pour vérifier que $X$ est une traverse minimale s'effectuent en $O$(1).
\end{enumerate}

Pour un hypergraphe donné $H$, l'algorithme \textsc{O-M2D} calcule donc l'ensemble des multi-membres en : $O$($m*n^{m+1}$ $\times$ $m$ $\times$ $C_n^k$) $\equiv$ $O$($m^2 * n^{m+1}$ * $C_n^k$).

\section{Etude expérimentale}
\label{expes}
Au cours de notre étude expérimentale, nous mettons l'accent sur une évaluation approfondie des performances des algorithmes présentés dans la section précédente. Nous comparons à travers les nombreuses expérimentations menées, les performances de \textsc{O-M2D} \textit{vs} respectivement \textsc{M2D}, \textsc{MtMiner} \cite{HebertBC07} et \textsc{ks} \cite{KavvadiasS05}. De tous les algorithmes d'extraction des traverses minimales existants dans la littérature, notre choix s'est porté sur ces deux derniers en raison des disponibilités de leurs codes sources \footnote[1]{Nous remercions les auteurs d'avoir mis à notre disposition leurs codes sources.}. Nous avons ainsi eu la possibilité de les modifier pour qu'ils ne calculent que les plus petites traverses minimales. Par ailleurs, tous les algorithmes considérés sont implémentés en \emph{C++} (compilés avec \emph{GCC} $4$.$1$.$2$) et les expérimentations réalisées sur une machine munie d'un processeur Intel Core $i7$ ayant une fréquence d'horloge de 2GHz et $6$ Go de mémoire centrale, et avec le système d'exploitation de Linux, \textsc{Ubuntu} $10$.$04$.

Durant ces expérimentations, nous avons considéré un jeu de données lié à une application de gestion de projet, des jeux de données "pire des cas" ainsi qu'un autre ensemble de jeux que nous avons construit à partir de deux bases de données du monde réel. Tout au long de notre étude expérimentale, nous avons vérifié que la borne maximale retournée par la fonction \textsc{GetMinTransversality} est bien égale au nombre de transversalité pour chaque hypergraphe traité.

\subsection*{Jeu de données de gestion de projets}
En gestion de projet, nous pouvons connaître les compétences requises pour mener à bien un projet donné, ainsi que celles des acteurs. L'objectif est alors d'identifier le plus petit ensemble d'individus capables de réaliser le projet. De plus, on peut souhaiter que chaque acteur puisse avoir le plus de compétences possibles. Ceci revient alors à chercher les \textsc{tmm}s.

Dans ce cas, une communauté serait composé d'un ensemble d'acteurs offrant une même compétence. Les communautés ne sont pas disjointes puisqu'un acteur peut avoir plusieurs compétences. Le jeu de données correspondant à ce problème est représenté par un hypergraphe constitué de 168 sommets, dont chacun correspond à un acteur, et 50 hyperarêtes, dont chacune correspond à une communauté, c'est à dire à une compétence nécessaire pour le projet.
\begin{table}[!h]
\begin{center}
\begin{tabular}{|r|r|r|r|r|r|r|}
\hline
 &\emph{$|$Comm.$|$} & \emph{$|$$\mathcal{X}$$|$} & \emph{$|$$\xi$$|$} & \emph{$|$$\mathcal{M}_H$$|$} & \emph{\#\textsc{Tmm}} &\emph{$\tau$(H)}\\
  \hline

   \textsc{\textbf{PM}} & 5 & 168   & 50    & 320 & 16 & 9 \\\hline

  \end{tabular}
\end{center}
\caption{Caractéristiques du jeu de données de gestion de projets}\label{AD}
\end{table}
Les caractéristiques de cet hypergraphe, appelé \textsc{PM}, sont résumés par le tableau \ref{AD} où \emph{$|$Comm.$|$} correspond à la taille de la plus petite hyperarête, \emph{$|$$\mathcal{X}$$|$} au nombre de sommets, \emph{$|$$\xi$$|$} au nombre d'hyperarêtes, \emph{$|$$\mathcal{M}_H$$|$} au nombre de traverses minimales, \emph{\#\textsc{Tmm}} au nombre de multi-membres calculés et \emph{$\tau$(H)} au nombre de transversalité. L'objectif est donc de rechercher les plus petits ensembles d'acteurs ayant les compétences requises pour mener à bien le projet.

\textbf{Performances et interpretations sur le jeu de données de gestion de projet}: Comme le montre le tableau \ref{exAD}, \textsc{MtMiner} est incapable de traiter ce jeu de données alors que les algorithmes \textsc{ks}, \textsc{M2D} et \textsc{O-M2D} nécessitent, respectivement, $307,15$, $1688$,$22$ et $158$,$93$ secondes pour extraire les traverses minimales multi-membres.

 Nos algorithmes ont extrait $320$ traverses minimales d'une taille égale à $9$ qui correspond à la valeur du nombre de transversalité, noté $\tau$($H$) dans la définition \ref{Deff}. Ceci signifie que nous devons réunir au moins neuf acteurs pour la réalisation du projet et ces $320$ traverses minimales correspondent aux sous-ensembles d'acteurs ayant les compétences nécessaires pour le projet. Parmi ces $320$ traverses minimales, seulement $16$ sont considérées comme des traverses minimales multi-membres. Ces dernières maximisent la condition de recouvrement. La particularité de ces traverses minimales multi-membres est qu'elles contiennent un ou plusieurs acteurs présentant diverses compétences. Ainsi, nos algorithmes parviennent à trouver des équipes, de taille minimale ayant le maximum de compétence, qui sont le plus aptes à conduire le projet.

\begin{table}[!h]
\begin{center}
\begin{tabular}{|c|r|r|r|r|}
\hline
              & \textsc{ks} &    \textsc{Mtminer} &  \textsc{M2D} & \textsc{O-M2D}\\
\hline
       \textsc{\textbf{PM}}   & 307,15 &    - &    1688,22 &       \textbf{158,93}\\
\hline

\end{tabular}
\end{center}
\caption{Jeu de données de gestion de projets: temps d'exécution (en secondes)}\label{exAD}
\end{table}

\subsection*{Bases de communautés sociales}
Dans cette seconde expérimentation, nous considérons des folksonomies, à partir desquelles nous avons extrait des communautés. Une \emph{folksonomie} est un néologisme, né de la jonction des mots \emph{folk} (\emph{i.e.,} les gens) et \emph{taxonomie}, désignant un système de classification collaborative par les internautes \cite{Mika2005aInf}. L'idée est de permettre à des utilisateurs de partager et de décrire des objets via des tags librement choisis. Formellement, une \emph{folksonomie} est composée de trois ensembles $\mathcal{U}$, $\mathcal{T}$, $\mathcal{R}$ et d'une relation ternaire $\mathrm{Y}$ entre eux, où $\mathcal{T}$ est un ensemble de tags (ou étiquettes) et $\mathcal{R}$ est un ensemble de ressources partagées par les utilisateurs, qui peuvent être des sites web à marquer\footnote[2]{\emph{http://del.icio.us}}, des vidéos personnelles à partager\footnote[3]{\emph{http://youtube.com}} ou des films à décrire\footnote[4]{\emph{http://movielens.org}} selon le type de la \emph{folksonomie} considérée. Quant à l'ensemble $\mathcal{U}$, il consiste en l'ensemble d'utilisateurs d'une \emph{folksonomie} qui sont décrits par leurs identifiants (pseudonymes).

Nous avons appliqué l'algorithme \textsc{Tricons}~\cite{Chi2012} pour l'extraction des tri-concepts associés à de telles folksonomies. Ces derniers sont des ensembles maximaux de la forme (Utilisateurs, Ressources, Tags): l'ensemble maximum d'utilisateurs, qui ont partagé un ensemble maximal de ressources qu'ils ont annoté avec un ensemble maximal de tags.

Pour extraire les communautés, nous projetons les tri-concepts sur la dimension "Utilisateurs". Ceci est réalisé en faisant varier le seuil du support minimal des utilisateurs, \emph{i.e.,} $minsupp_{u}$, qui est le nombre minimal d'utilisateurs qu'un tri-concept peut contenir. Dans ce qui suit, nous décrivons les folksonomies considérées au cours de nos expérimentations.

\begin{enumerate}
\item \textsc{del.icio.us\footnote{www.delicious.com}}: le système \textsc{del.icio.us} est un service de marque-page social qui offre à ses utilisateurs la possibilité de partager leurs pages web préférées. La base de données considérée dans ce rapport contient tous les marque-pages ajoutés sur le site \textit{http://delicious.com} en Janvier $2007$. Le processus de récupération regroupe quelque $494, 636$ marque-pages qui ont été publiés par $54,915 $ utilisateurs par le biais de $64,968 $ tags sur $129,220 $ ressources. Dans cette étude expérimentale, nous considérons qu'une communauté, dans une base \textsc{del.icio.us}, est constituée des utilisateurs ayant partagé, au moins, deux mêmes pages web. Avant l'application de nos algorithmes, un pré-traitement sur ces données a permis de dégager les communautés qui serviront d'hyperarêtes dans l'hypergraphe d'entrée.

\item \textsc{MovieLens\footnote{www.movielens.umn.edu}}: il s'agit d'un système de recommandation filmographique \textsc{MovieLens}, dont le site web a été conçu par un groupe de recherche, \emph{GroupLens}, à l'université de Minnesota, aux États-Unis. Disponible au public, ce jeu de données contient des évaluations explicites au sujet de films. Le site met à disposition deux jeux de données d'évaluations de films, de tailles différentes. Le premier jeu comprend $1,000,000$ évaluations, de 1 à 5 étoiles, faites par environ $6,000$ utilisateurs, et le second comprend $100,000$ évaluations fournies par $943$ utilisateurs sur $1,682$ films, entre Septembre $1997$ et Avril $1998$. C'est ce second jeu de données que nous avons utilisé pour nos tests, en considérant qu'une communauté est formée des utilisateurs qui ont fourni leur avis sur au moins deux mêmes films.

\end{enumerate}

En variant $minsupp_{u}$, nous obtenons quatre jeux de données à partir de la base de données \textsc{del.icio.us} des folksonomies (notés \emph{Del1}, \emph{Del2}, \emph{Del3} et \emph{Del4}) et trois jeux de données à partir de la base de données \textsc{Movielens} des folksonomies (notés \emph{Mov1}, \emph{Mov2} et \emph{Mov3}). Dans l'hypergraphe associé à chaque jeu de données, les sommets représentent les utilisateurs et les hyperarêtes correspondent aux communautés où une communauté représente un ensemble d'utilisateurs ayant partagé le même ensemble de ressources avec le même ensemble de tags. Pour chaque jeu de données, l'objectif est de trouver le plus petit ensemble d'utilisateurs permettant de représenter toutes les communautés.

Les caractéristiques des différents jeux de données sont résumés dans le Tableau \ref{carac}. Ainsi, la première colonne \emph{$\mid$Comm.$\mid$} indique le nombre minimum de sommets dans une communauté (\textit{i.e.} $minsupp_{u}$). La seconde colonne contient le nombre de sommets (\emph{$|$$\mathcal{X}$$|$}) de l'hypergraphe et la troisième le nombre d'hyperarêtes (\emph{$|$$\xi$$|$}). La quatrième colonne indique le nombre de traverses minimales (\emph{$|$$\mathcal{M}_H$$|$}) que renferme l'hypergraphe. L'avant-dernière colonne montre le nombre de \textsc{Tmm}s. La dernière colonne correspond à la taille des \textsc{Tmm}s, en termes de nombre de sommets ($\tau$(H)). Nous pouvons noter que plus la valeur de \emph{$\mid$Comm.$\mid$} est basse, plus le nombre de \textsc{Tmm}s (\emph{\#\textsc{Tmm}}) et leurs tailles (\emph{$\tau$(H)}) sont élevées, atteignant 21 pour \emph{Del4} et 20 pour \emph{Mov3}.

\begin{table}[htbp]
\begin{center}
\centering
\begin{tabular}{|r|r|r|r|r|r|r|}
\hline

   &\emph{$|$Comm.$|$} & \emph{$|$$\mathcal{X}$$|$} & \emph{$|$$\xi$$|$} & \emph{$|$$\mathcal{M}_H$$|$} & \emph{\#\textsc{Tmm}} & \emph{$\tau$(H)} \\
  \hline

    \textbf{Del1} & 6 & 51   & 38    & 13    & 4 & 5 \\\hline
    \textbf{Del2} & 5 & 119   & 91    & 52    & 10 & 6 \\\hline
    \textbf{Del3} & 3 & 165   & 157   & 1800  & 78 & 13 \\\hline
    \textbf{Del4} & 2 & 248   & 179   & 8976  & 201 & 21\\\hline

  \hline \hline \hline

    \textbf{Mov1} & 5 & 88   & 80    & 108 & 1 & 6\\\hline
    \textbf{Mov2} & 3 & 143   & 246  & 172  & 3 & 12\\\hline
    \textbf{Mov3} & 2 & 196   & 501   & 306  & 26 & 20\\\hline

  \end{tabular}

\end{center}
\caption{Caractéristiques des bases sociales [(\textbf{Haut}) \textsc{del.icio.us} (\textbf{Bas}) \textsc{MovieLens}]}\label{carac}
\end{table}

\textbf{Performances}: comme le montre le tableau \ref{del}, les différents tests confirment que l'algorithme \textsc{O-M2D} surpasse largement les algorithmes \textsc{M2D}, \textsc{MtMiner} et \textsc{ks}, pour l'ensemble des jeux de données. Par ailleurs, si \textsc{M2D} présente des temps d'exécution élevés, il est assez robuste pour venir à bout de tous les jeux de données, alors que la consommation mémoire élevée de \textsc{MtMiner} l'empêche de s'exécuter sur les jeux \emph{Del3}, \emph{Del4} et \emph{Mov3}.
En déterminant le nombre d'éléments dans une \textsc{Tmm}, \textsc{M2D} est capable d'élaguer de nombreux candidats qui sont générés et traités par \textsc{MtMiner}.
Dans les deux types de jeux de données, les résultats confirment, par ailleurs, que l'écart, en termes de temps d'exécution, des différents algorithmes est en faveur de \textsc{O-M2D}. Cet écart est plus conséquent quand les nombres de sommets et d'hyperarêtes sont grands.
En effet, la dernière colonne du tableau \ref{carac} montre que la taille des \textsc{Tmm}s est très large, atteignant respectivement, 13 pour \emph{Del3}, 21 pour \emph{Del4}, et 20 pour \emph{Mov3}. L'avantage principal de \textsc{O-M2D} se résume dans sa faculté à cibler directement ce niveau (appelé \textit{level} dans l'Algorithm 2 et $\tau$(H) dans la \textit{Définition} \ref{Deff}).

Par exemple, pour \emph{Del3} et \emph{Del4}, \textsc{MtMiner} ne peut pas extraire les traverses minimales quand l'hypergraphe renferme plus de $157$ hyperarêtes. Par ailleurs, sachant que le nombre de traverses minimales croît exponentiellement, l'avantage que présente l'algorithme \textsc{O-M2D} est sa capacité à fournir un résultat sans stocker les candidats en mémoire.
\begin{table}[!h]
\begin{center}
\begin{tabular}{|c|r|r|r|r|}
\hline
              & \textsc{ks} & \textsc{Mtminer} &  \textsc{M2D} & \textsc{O-M2D}\\
\hline
       \textbf{Del1}   &  88,26  &  77,45 &    276,65 &       \textbf{61,28}\\
\hline
       \textbf{Del2}  &  263,90 &  200,66 &   964,32  &      \textbf{112,50}\\
\hline
       \textbf{Del3} &   401,38  &   -  &    1920,12 &     \textbf{174,78}\\
\hline
       \textbf{Del4} &    793,08  &    - &   2880,84  &      \textbf{364,92} \\ \hline

\hline \hline \hline

       \textbf{Mov1}    &  72,00  & 53,28 &   335,71 &       \textbf{59,52}\\
\hline
       \textbf{Mov2}    & 262,09 & 185,56 &     1492,34 &      \textbf{131,84}\\
\hline
       \textbf{Mov3}   & 881,73 &       -  &   2655,63 &      \textbf{351,27}\\
\hline
\end{tabular}
\end{center}
\caption{Bases sociales [(\textbf{Haut}) \textsc{del.icio.us} (\textbf{Bas}) \textsc{MovieLens}]: Temps d'exécution (en secondes)}\label{del}
\end{table}

Dans le but d'analyser, en profondeur, les \textsc{Tmm}s calculées par nos algorithmes, considérons le jeu de données \emph{Del2} du Tableau \ref{carac} et les caractéristiques des sommets qui appartiennent à l'ensemble des \textsc{Tmm}s.

\begin{figure}[!h]
\begin{center}
\includegraphics[scale=0.65]{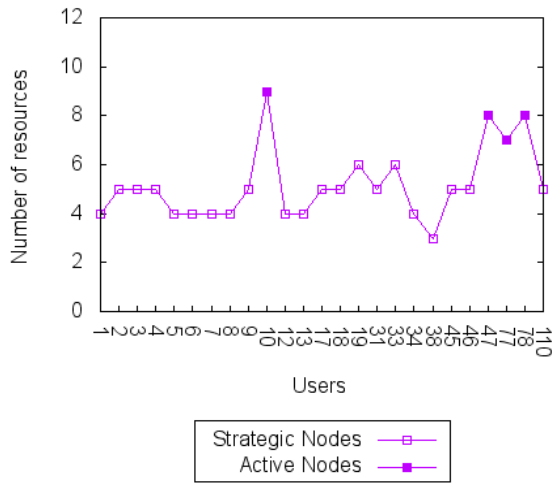}
\caption{Nombre de ressources partagées par les $25$ utlisateurs les plus actifs}\label{ress}
\end{center}
\end{figure}

\textsc{O-M2D} donne en sortie \emph{10} \textsc{Tmm}s de taille 6, i.e, composé de 6 sommets. Cela signifie que nous devons trouver au moins 6 utilisateurs pour représenter l'ensemble des communautés. Un examen de près de ces \textsc{Tmm}s montre que \emph{4} sommets (10, 47, 77, 78) appartiennent à tous les \textsc{Tmm}s extraites à partir de ce jeu de données. Nous les appelons "\textit{actifs}" car ils ont la plus importante activité de marquage (tagging) dans le jeu de données Del2 de \textsc{del.icio.us}, comme le montre la Figure \ref{ress} et la Figure \ref{tags}.

Les deux autres sommets (que nous appelons "\textit{stratégiques}") n'ont pas, au contraire, une activité de marquage exceptionnelle. Leur appartenance aux \textsc{Tmm}s s'explique par le fait qu'ils représentent des communautés qui ne renferment aucun sommet actif.

\begin{figure}[!h]
\begin{center}
\includegraphics[scale=0.65]{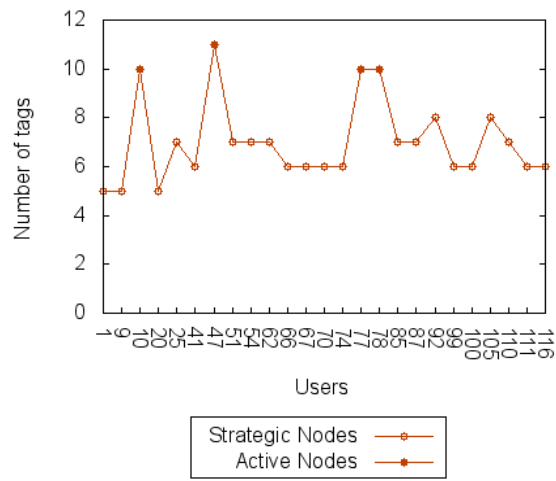}
\caption{Nombres de tags des $25$ utilisateurs les plus actifs}\label{tags}
\end{center}
\end{figure}

\textbf{Consommation mémoire}: les statistiques fournies par le Tableau \ref{Mdel} mettent en évidence la consommation en \textsc{ram} très faible des algorithmes \textsc{O-M2D} et \textsc{ks}. Lorsque \textsc{MtMiner} et \textsc{M2D} doivent générer et sauvegarder en mémoire tous les candidats de tous les niveaux balayés, \textsc{O-M2D} cible directement le niveau adéquat et teste la condition d'essentialité des candidats générés. Un examen attentif du nombre de transversalité des hypergraphes \emph{Del4} et \emph{Mov3}, dans le Tableau \ref{carac}, indique la quantité de candidats que \textsc{MtMiner} et \textsc{M2D} ont à traiter. En effet, ces derniers algorithmes doivent atteindre le $21^{ème}$ niveau pour le jeu de données Del4 et le $20^{ème}$ niveau pour Mov3. Ceci explique pourquoi \textsc{MtMiner} est dans l'incapacité d'atteindre le $20^{ème}$ niveau.

\begin{table}[!h]
\begin{center}
\begin{tabular}{|c|r|r|r|r|}
\hline
              & \textsc{ks} &  \textsc{Mtminer} & \textsc{M2D}& \textsc{O-M2D}\\
\hline
       \textbf{Del1}   & 4.335  &  4.334.102  &    1.807.980 &       \textbf{9.223}\\
\hline
       \textbf{Del2}  & 6.290  & 5.568.931 &    2.606.722  &      \textbf{11.810}\\
\hline
       \textbf{Del3}   & 9.603 &       -  &    3.723.119 &      \textbf{16.369}\\
\hline
       \textbf{Del4}   & 13.844  &       - &    4.458.656  &      \textbf{18.454}\\\hline

\hline \hline \hline

       \textbf{Mov1}  & 3.991  &    3.518.223 &    1.366.841 &       \textbf{8.604} \\
\hline
       \textbf{Mov2}  & 6.387  &    4.976.213 &    2.461.857 &      \textbf{10.611} \\
\hline
       \textbf{Mov3}  & 9.640  &          - &    3.004.886  &      \textbf{13.602} \\
\hline
\end{tabular}\end{center}
\caption{Bases sociales[(\textbf{Haut}) \textsc{del.icio.us} (\textbf{Bas}) \textsc{MovieLens}]: Consommation mémoire (en KO)}\label{Mdel}
\end{table}

\subsection*{Bases "pires des cas"}

Les bases "pire des cas" sont introduites pour étudier plus en profondeur les performances des algorithmes considérés au cours de cette étude expérimentale. Elles correspondent à une matrice d'incidence définie comme suit :
\begin{definition}\label{wc}
Un contexte "pire des cas" $IM_H = (\xi,\mathcal{X},\mathcal{R})$, où $\xi$ et $\mathcal{X}$ sont, respectivement, les ensembles finis d'hyperarêtes et de sommets de l'hypergraphe $H$, est une matrice dans laquelle toutes les hyperarêtes sont formées du même nombre d'éléments, égal à $n$, et où chaque sommet a un support égal à 1 ($Supp(x)$ = 1, $\forall x \in \mathcal{X}$).\\
\end{definition}

Par exemple, un contexte "pire des cas" pour $|\xi|=3$, $\mathcal{X}=9$ et $n=3$, est donné par le Tableau \ref{pire1}.
Les bases "pire des cas" nous permettent d'évaluer le comportement des algorithmes dans des cas extrêmes. Le test consiste à varier les valeurs de $n$ et $|\xi|$, jusqu'à ce que les algorithmes ne puissent plus s'exécuter correctement.
\begin{table}[!h]
 \centering

\begin{tabular}{|r|r|r|r|r|r|r|r|r|r|}
\hline
    {\bf } &   {\bf $x_1$} &   {\bf $x_2$} &   {\bf $x_3$} &   {\bf $x_4$} &   {\bf $x_5$}&   {\bf $x_6$}&   {\bf $x_7$}&   {\bf $x_8$}&   {\bf $x_9$}\\
\hline
  {\bf $e_1$} &     $\times$ &    $\times$        &     $\times$       &        &            &            &         &            &                        \\
\hline
 {\bf $e_2$} &      &            &            &    $\times$   &      $\times$        &    $\times$       &        &             &                        \\
\hline
 {\bf $e_3$} &     &            &             &        &            &            &    $\times$    &      $\times$      &       $\times$                  \\
\hline
\end{tabular}

 \caption{Base pire des cas pour $|\xi|=3$}\label{pire1}
\end{table}



\begin{table}[!h]
\begin{center}
\begin{tabular}{|c|c||r|r|r|r|}

\hline
      $n$ &          $|\xi|$  & \textsc{ks} &  \textsc{Mtminer} &  \textsc{M2D} & \textsc{O-M2D} \\

\hline
         4 &         11  & 174,06  &   180,23 &     257,97 &        87,54 \\
\hline
         4 &         12  & 215,73 & - &     295,43 &        111,21 \\
\hline
         4 &         19  & 498,68 & - &     1457,78 &        380,89 \\
\hline
         4 &         73  & 5527,90 & - &     - &        3234,30 \\
\hline \hline
         5 &         10  &  194,42 & 214,31 &      308,10&      90,12 \\
\hline
         5 &         11   & 229,14 & - &      352,10&      96,12 \\
\hline
         5 &         19  & 650,97 & - &      1866,42&      398,78 \\
\hline
         5 &         60  & 5139,02 & - &      - &      3094,17 \\
\hline \hline
         6 &          9  & 211,30  & 245,12 &      344,64&          97,38 \\
\hline
         6 &          10 & 228,74 &  - &      400,58&          101,09 \\
\hline
         6 &          16 & 616,51  & - &      2119,47&          472,39 \\
\hline
         6 &          53 & 5349,28  & - &      - &          3181,40 \\
\hline \hline
         7 &          9  & 248,59 & 299,46 &      400,49 &         100,72 \\
\hline
         7 &          10 &  262,93 &- &      468,18 &         119,72 \\
\hline
         7 &          16 & 668,31 & - &      2288,46 &         495,08 \\
\hline
         7 &          44 & 5172,55  &- &      - &         2802,89 \\
\hline \hline
         8 &          8  & 257,30 & 326,02 &      466,21 &        103,69 \\
\hline
         8 &          9  & 284,36  & - &      517,36 &        138,28 \\
\hline
         8 &          14 & 700,94 & - &      2557,73 &        531,44 \\
\hline
         8 &          29 & 5311,80 & - &      - &        2949,06 \\
\hline \hline
         9 &          7  & 248,27 & 387,66 &      511,92 &        104,09 \\
\hline
         9 &          8  & 270,686 & - &    534,84   &        168,03 \\
\hline
         9 &          13 & 753,04 & - &      2780,79 &        567,56 \\
\hline
         9 &          22 & 4868,93 & - &      - &        2354,98 \\
\hline \hline
         10 &          7 & 264,91 & 428,88 &      683,36 &         110,34 \\
\hline
         10 &          8 & 289,53 & - &      711,49 &         133,28 \\
\hline
         10 &          11  & 841,71 & - &      3283,36 &         624,66 \\
\hline
         10 &          13  & 4390,22 & - &      - &        1899,29 \\
\hline
\end{tabular}
\end{center}
  \caption{Bases pire des cas: Temps d'exécution (en secondes)}\label{worstdb}
\end{table}
\textbf{Performances}: le Tableau \ref{worstdb} montre pour chaque valeur de $n$, le nombre maximal d'hyperarêtes qui peuvent être traitées par les algorithmes considérés et les temps de traitement associés.
Grâce à sa capacité à déterminer le nombre d'éléments d'une \textsc{Tmm}, l'algorithme \textsc{O-M2D} présente un avantage indéniable par rapport à \textsc{MtMiner} et \textsc{M2D}. Selon les données du Tableau \ref{worstdb}, pour une valeur donnée de $n$, \textsc{O-M2D} présente des temps d'exécution nettement meilleurs que les trois autres algorithmes et est capable de traiter des bases "pire des cas" pour des valeurs élevées de $|\xi|$. A titre d'exemple, pour $n$ = 4, \textsc{MtMiner} s'arrête brusquement pour un nombre d'hyperarêtes égal à 12 alors que l'algorithme \textsc{M2D} résiste mieux et s'arrête à une valeur égale à 20. Dans le même temps, \textsc{O-M2D} et \textsc{ks} s'arrête pour une valeur de $|\xi|$ égal à $74$.

\begin{table}[!h]
\begin{center}
\begin{tabular}{|c|c||r|r|r|r|}
\hline
        $n$ &          $|\xi|$  & \textsc{ks}&    \textsc{MtMiner}  &    \textsc{M2D} & \textsc{O-M2D} \\
\hline
         4 &         11  & 4.285  &2.682.886 &  1.398.441 &      6.108 \\
\hline
         4 &         12  &  4.300 & - &  1.844.364 &      6.221 \\
\hline
         4 &         19  &  4.617 &- &  4.995.732 &      6.908 \\
\hline
         4 &         73  &  20.093 &- &  - &      32.805 \\
\hline \hline
         5 &         10  & 3.866  &2.811.429 &  1.470.266 &      5.294 \\
\hline
         5 &         11  & 4.077  &- &  1.719.498 &      5.565 \\
\hline
         5 &         19  &  4.902 &- &  5.338.173 &      6.403 \\
\hline
         5 &         60  & 18.620  &- &  - &      26.164 \\
\hline \hline
         6 &         9  & 3.996  &3.066.422 &  1.712.089 &      5.482 \\
\hline
         6 &         10 &  4.098 &- &  1.999.695 &      5.607 \\
\hline
         6 &         16 & 4.781  &- &  5.514.683 &      6.962 \\
\hline
         6 &         53 & 15.092  &- &  - &      22.197 \\
\hline \hline
         7 &          9 & 4.094  &3.185.089 &  2.085.366 &      5.165 \\
\hline
         7 &          10  &  4.168 &- &  2.473.281 &      5.537 \\
\hline
         7 &          16  &  4.830 &- &  5.800.962 &      6.308 \\
\hline
         7 &          44  & 10.982  &- &  - &      17.389 \\
\hline \hline
         8 &          8  & 3.841  &3.541.797 &  2.226.625 &      4.821 \\
\hline
         8 &          9  &  4.117 &- &  2.826.625 &      4.996 \\
\hline
         8 &          14 & 4.497  &- &  5.741.793 &      5.570 \\
\hline
         8 &          29 & 8.896  &- &  - &      11.411 \\
\hline \hline
         9 &          7  & 3.602  &3.922.008 &  2.677.026 &      4.757 \\
\hline
         9 &          8  & 3.997  &- &  3.168.442 &      4.886 \\
\hline
         9 &          13 &  4.406 &- &  5.694.223 &      5.101 \\
\hline
         9 &          22 & 7.911  &- &  - &      9.734 \\
\hline \hline
        10 &          7  &  3.802 &4.433.787 &  3.019.860 &      4.955 \\
\hline
        10 &          8  &  4.106 &- &  3.840.117 &      5.093 \\
\hline
        10 &          11 &  4.212 &- &  5.899.004 &      5.202 \\
\hline
        10 &          13 &  4.537 &- &  - &      5.817 \\
\hline
\end{tabular}
\end{center}
\caption{Bases pire des cas: Consommation mémoire (en KO)}\label{Mwc}
\end{table}

\textbf{Consommation mémoire}: pour ces bases "pire des cas", nous avons étudié aussi la consommation mémoire. Ainsi, le Tableau \ref{Mwc} montre, respectivement, la mémoire consommée par les cinq algorithmes. \textsc{O-M2D} et \textsc{ks} se montrent alors très performant par rapport aux algorithmes \textsc{MtMiner} et \textsc{M2D}. En effet, par rapport à la consommation mémoire des trois autres algorithmes, celle de \textsc{O-M2D}, tout comme celle de \textsc{ks}, est négligeable. \textsc{O-M2D} ne stocke en mémoire que l'hypergraphe d'entrée et les candidats, de taille égale au nombre de transversalité, générés sont traités sans être sauvegardés.

\section{Conclusion}
\label{concl}
Au cours de ce chapitre, nous avons introduit une nouvelle approche pour la détection d'une classe particulière des traverses minimales, que nous avons appelées traverses minimales multi-membres, à partir d'un système communautaire représenté par un hypergraphe. L'une de nos contributions se trouve dans la définition des \textsc{Tmm}s en se basant sur la notion d'ensemble de sommets essentiels. Ceci nous a permis de mettre en place un algorithme optimisé, qui cible directement le niveau qui renferme les \textsc{Tmm}s. Des expérimentations effectuées sur différents jeux de données ont montré que l'algorithme \textsc{O-M2D} présente des performances très intéressantes par rapport à celles obtenues avec des algorithmes classiques. Cette contribution a été publiée dans une revue internationale \cite{NidhalJSS} et deux conférences avec comité de lecture \cite{NidhalEGC2012, NidhalCAP2012}. Dans le chapitre suivant, nous allons étendre notre approche pour extraire l'ensemble de toutes les traverses minimales, en proposant une représentation concise et exacte de cet ensemble grâce à la notion d'irrédondance. Ceci a été motivé par le nombre de traverses minimales, qui peut être exponentiel même pour des hypergraphes simples.

\chapter{"Diviser pour régner" pour l'extraction des traverses minimales d'un hypergraphe}\label{chap5}
\section{Introduction}
Optimiser le calcul des traverses minimales en optant pour la décomposition de l'hypergraphe d'entrée peut présenter une solution intéressante à condition de bien choisir le nombre optimal d'hypergraphes partiels de manière à éliminer des tests de minimalité des traverses calculées. En se basant sur le nombre de transversalité, nous proposons une approche basée sur la stratégie "diviser pour régner" pour l'extraction des traverses minimales. Un hypergraphe peut, en effet, être décomposé en un certain nombre d'hypergraphes partiels, égal à la taille de la plus petite traverse minimale que renferme l'hypergraphe d'entrée. Les traverses minimales extraites à partir des hypergraphes partiels, que nous appellerons "locales", permettent de retrouver l'ensemble de toutes les traverses minimales de l'hypergraphe initial. Les traverses obtenues et dont la cardinalité est égale au nombre de transversalité seront considérées directement comme des traverses minimales et ce, sans vérifier par des tests, leurs minimalités. C'est ce que nous détaillons à travers ce chapitre via l'algorithme \textsc{local-generator} et à travers l'étude expérimentale dont il a fait l'objet.

\section{Objectifs de la décomposition}\label{intro}
La principale difficulté que pose l'extraction des traverses minimales réside dans le nombre exponentiel de ces dernières, même quand l'hypergraphe d'entrée est simple, comme le montre l'exemple \ref{expnb} du chapitre \ref{chap4} (page \pageref{expnb}).

\subsection{Diviser pour régner}

Les algorithmes d'extraction des traverses minimales les plus performants \cite{BMA03,KavvadiasS05,uno1} sont des améliorations de l'algorithme de Berge \cite{berge1989}. Ce dernier traite les hyperarêtes une à une en calculant à chaque itération \emph{i} les traverses minimales de l'hypergraphe constitué par les \emph{i-èmes } hyperarêtes considérées. Avec pour objectif d'optimiser le calcul des traverses minimales, notre approche repose sur cette idée en usant du paradigme "\emph{diviser pour régner}", présenté dans la Définition \ref{divpourreg}.

\begin{definition}\textsc{Diviser pour régner}\label{divpourreg}
Le paradigme \emph{diviser pour régner} se compose de trois étapes. La première est de \emph{diviser} le problème en un certain nombre de sous-problèmes. La deuxième est de \emph{régner} sur ces sous-problèmes en les résolvant de manière récursive ou directement. Enfin, \emph{combiner} les solutions des sous-problèmes en une solution finale du problème initial.
\end{definition}

Le principe consiste à réduire ce nombre d'itérations en décomposant l'hypergraphe en un nombre précis d'hypergraphes partiels, équivalent au nombre de transversalité de l'hypergraphe d'entrée $H$. A partir de chaque hypergraphe partiel $H_i$, nous calculons alors ce que nous appelons \emph{les traverses minimales locales} à $H_i$. Le produit cartésien de ces traverses minimales locales correspondra alors à un ensemble de traverses de $H$ qui seront soumises à une vérification de la minimalité pour être considérées comme des traverses minimales. En outre, pour un hypergraphe $H$ avec un nombre de transversalité égal à $k$, le fait de décomposer $H$ en $k$ hypergraphes $H_i$ permet d'éliminer le test de la minimalité pour les ensembles de sommets de taille $k$ qui seront considérés comme traverses minimales de $H$, sans aucun autre calcul supplémentaire.

Il est clair que cette approche ne saurait être efficace sur des hypergraphes, dont le nombre de tranversalité est très bas dans la mesure où le nombre d'hypergraphes partiels n'est pas conséquent et ne permet pas une optimisation intéressante du calcul des traverses minimales. Un profil du type d'hypergraphes sur lequel notre approche peut se montrer efficace est dressé au terme de l'étude expérimentale que nous détaillons à la fin de ce chapitre.


\subsection{Originalité de l'approche}\label{art}
Nous avons détaillé dans le chapitre $2$ les différents algorithmes dédiés à l'extraction des traverses minimales. Nous avons souligné que les algorithmes les plus performants étaient des améliorations de l'algorithme de Berge mais ils ne reposent pas sur le paradigme "diviser pour régner". En ce sens, notre approche est une extension complètement différente de l'algorithme de Berge. Alors que dans ce dernier, ainsi que dans les améliorations qui en ont été proposées, l'idée est de traiter les hyperarêtes une à une, nous nous proposons de traiter les hyperarêtes ensemble par ensemble. L'hypergraphe d'entrée $H$ se trouve alors décomposé en un nombre d'hypergraphes partiels égal au nombre de transversalité $k$ de $H$.

Chaque hypergraphe partiel renferme des traverses minimales locales et le produit cartésien, combiné à un test de la minimalité, permet de retrouver l'ensemble des traverses minimales de $H$. Le test de la minimalité est nécessaire dans la mesure où les traverses minimales locales sont effectivement minimales mais au sein de l'hypergraphe partiel à partir duquel elles sont extraites mais rien n'assure qu'elles seront minimales dans l'hypergraphe initial. Comme nous l'avons déjà souligné, seules les traverses de taille égale au nombre de transversalité sont effectivement minimales puisqu'il ne peut exister une traverse incluse dedans.

Un seul algorithme, parmi tout ceux présentés dans le chapitre $2$ met à profit cette notion de décomposition mais d'une façon différente. Il s'agit de celui de Bailey \textit{et al.}, qui consiste à décomposer les hyperarêtes formées par un nombre important de sommets de manière à n'avoir que des hyperarêtes de relativement petite taille. Dans notre algorithme, l'hypergraphe est décomposé indépendamment de la taille de ses hyperarêtes.

Le nombre de transversalité d'un hypergraphe est la notion-clé, autour de laquelle est bâtie notre approche. Le choix du nombre d'hypergraphes partiels n'est pas arbitraire puisqu'il garantit que les traverses, dont la taille est égale à $k$, peuvent être directement considérées comme des traverses minimales de $H$. C'est sur cette élimination de certains de ces tests inutiles de la minimalité que nous comptons pour optimiser les temps d'extraction des traverses minimales.

\section{Définitions et notations}
Dans cette section, nous proposons de présenter des définitions clés et notations que nous utiliserons tout au long des sections suivantes. Pour aboutir à notre approche d'extraction des traverses minimales, basée sur la notion de traverse minimale locale, nous avons encore combiné des concepts de la théorie des hypergraphes (union et produit cartésien d'hypergraphes) avec d'autres de la fouille de données (ensemble essentiel, support), présentés dans les chapitres précédents.

\begin{definition} \textsc{Union et produit cartésien} \cite{berge1989} \label{def1}\\
Soit $H$ = ($\mathcal{X}$, $\xi$) et $G$ = ($\mathcal{X'}$, $\xi'$) deux hypergraphes tels que $\xi$ = $\{\xi_1, \xi_2, \ldots, \xi_m\}$ et $\xi'$ = $\{ \xi'_1, \xi'_2, \ldots, \xi'_{m'}\}$.

$H$ $\cup$ $G$ représente l'union de $H$ et $G$. Le résultat de cette union est un hypergraphe dont l'ensemble des sommets est constitué de ceux de $H$ et de $G$, et l'ensemble des hyperarêtes contient celles de $H$ et $G$, qui par souci de simplification sera aussi noté $H$ $\cup$ $G$:

$H$ $\cup$ $G$   = $(\mathcal{X} \cup \mathcal{X'}, \xi \cup \xi')$

$H \times G$ représente le produit cartésien des deux hypergraphes dont le résultat est un hypergraphe dont l'ensemble des sommets contient ceux des deux hypergraphes. Quant à l'ensemble des hyperarêtes, il est aussi noté $H \times G$ et est égal au produit cartésien de $\xi$ et de $\xi'$ autrement dit à l'union de tous les couples possibles d'hyperarêtes tels que le premier élément appartient à $\xi$ et le deuxième à $\xi'$ :

$H \times G$ =  $(\mathcal{X} \cup \mathcal{X'}, \{(\xi_i \cup \xi'_j), i = 1, \ldots, m, j= 1, \ldots, m')$

\end{definition}


\begin{proposition} \cite{berge1989}\label{KSprop}\\
Soient $H$ et $G$, deux hypergraphes simples. Les traverses minimales de l'hypergraphe $H$ $\cup$ $G$ sont des couples, minimaux au sens de l'inclusion, générés par le produit cartésien des ensembles de traverses minimales de $H$ et de $G$ :

  $\mathcal{M}_{H \cup G}$ = Min\{ $\mathcal{M}_H$ $\times$ $\mathcal{M}_G$\}.
\end{proposition}

\begin{definition} \textsc{Hypergraphe partiel} \cite{berge1989}\\
Un hypergraphe partiel $H'$ est la restriction d'un hypergraphe $H$ à un sous-ensemble d'hyperarêtes $\xi'$ incluses dans $\xi$ et aux sommets contenus dans ces hyperarêtes.
\end{definition}

Dans le cadre de notre approche, nous proposons d'étendre la proposition \ref{KSprop} en considérant plus de deux hypergraphes. Plus précisément, à partir d'un hypergraphe $H$=($\mathcal{X}$, $\xi$), dont le nombre de transversalité $\tau$($H$) est égal à $k$, et d'une traverse minimale $T$ = $\{x_1, x_2, \ldots, x_k\}$  de  $\mathcal{M}_H$ de taille $k$ dont les sommets sont triés par ordre de support décroissant de sorte que $x_1$ est le sommet qui appartient au plus grand nombre d'hyperarêtes, nous proposons de construire $k$ hypergraphes partiels $H_i$= ($\mathcal{X}_i$, $\xi_i$), $i =1$, \ldots, $k$ tels que :\\


\begin{itemize}
 \item $\xi_1$ = $\{e \in \xi \mid x_1 \in  e\}$
\item $\mathcal{X}_1$ = \{$x$ $\in$ $\mathcal{X} \mid x \in e, \forall e \in \xi_1$ $\}$
\item ..
\item ..
\item $\xi_i$ = $\{e \in \xi - \bigcup\limits_{j=1}^{i-1} \xi_{j} \mid x_i \in  e\}$
\item $\mathcal{X}_i$ = \{$x$ $\in$ $\mathcal{X} \mid x \in e, \forall e \in \xi_i$ $\}$
\end{itemize}
On peut remarquer que les hypergraphes partiels $H_i$ vérifient de façon évidente les propriétés suivantes :
\begin{itemize}
 \item $\xi_i \subseteq \xi$
  \item $\bigcup\limits_{i=1}^k \xi_{i}$ = $\xi$.
  \item $\nexists e \in \xi$ tel que $e \in \xi_i \cap \xi_j$, $i \neq j$.
\end{itemize}

Les traverses minimales de l'hypergraphe partiel $H_i$ sont appelées \emph{traverses minimales locales} à $H_i$ et leur ensemble est noté par $\mathcal{M}_{H_i}$.

\begin{figure}[htbp]
\begin{minipage}{.4\textwidth}\centering
\includegraphics[scale=0.7]{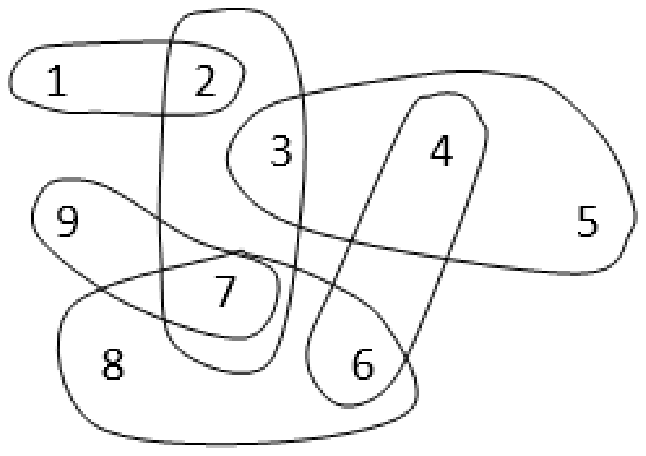}
\end{minipage}\hfill
\begin{minipage}{.6\textwidth}\centering
\small{
\begin{tabular}{|r|r|r|r|r|r|r|r|r|r|}
\hline
    {\bf } &   {\bf $\textbf{1}$} &   {\bf $\textbf{2}$} &   {\bf $\textbf{3}$} &   {\bf $\textbf{4}$} &   {\bf $\textbf{5}$} &   {\bf $\textbf{6}$} &   {\bf $\textbf{7}$} &   {\bf $\textbf{8}$}&   {\bf $\textbf{9}$}\\
\hline
  {\bf $e_1=\textbf{\{1, 2\}}$} &     $1$ &       $1$     &       0     &  0  &0 &0  & 0 &  0   & 0     \\
\hline
  {\bf $e_2=\textbf{\{2, 3, 7\}}$} &     0 &       $1$     &       $1$     &  0 &0 &0  & $1$ &  0   & 0 \\
\hline
  {\bf $e_3=\textbf{\{3, 4, 5\}}$} &     0 &       0     &       $1$     &  $1$  &$1$ & 0  & 0 &  0   & 0  \\
\hline
  {\bf $e_4=\textbf{\{4, 6\}}$} &     0 &       0     &       0     &  $1$  & 0 & $1$  & 0 &  0   & 0    \\
  \hline
    {\bf $e_5=\textbf{\{6, 7, 8\}}$} &     0 &       0     &       0     &  0  &0 & $1$  & $1$ &  $1$   & 0    \\
  \hline
    {\bf $e_6=\textbf{\{7, 9\}}$} &     0 &       0     &       0     &  0  &0 &0  & $1$ &  0   & $1$    \\
  \hline
\end{tabular}}
\end{minipage}
\caption{Un exemple d'hypergraphe $H = (\mathcal{X}, \xi)$ et la matrice d'incidence $IM_H$ correspondante}\label{hyp1}
\end{figure}

\begin{example}
La figure \ref{hyp1} illustre un hypergraphe simple $H = (\mathcal{X}, \xi)$ tel que $\mathcal{X} = \{1, 2, 3, 4, 5, 6, 7, 8, 9\}$ et $\xi$ = $\{$ $e_1$, $e_2$, $e_3$, $e_4$, $e_5$, $e_6$ $\}$ avec $e_1$ = $\{$1, 2$\}$, $e_2$ = $\{$2, 3, 7$\}$, $e_3$ = $\{$3, 4, 5$\}$, $e_4$ $\{$4, 6$\}$, $e_5$ = $\{$6, 7, 8$\}$ et $e_6$ = $\{$7, 9$\}$. $H$ a un nombre de transversalité égal à $3$. $H$ possède $2$ traverses minimales de cardinalité minimale égale à $3$ : $\{1, 4, 7\}$ et $\{2, 4, 7\}$. Prenons, par exemple, la traverse minimale $\{1, 4, 7\}$. Après avoir ordonné les trois sommets le composant, selon un ordre décroissant de support, nous obtenons les trois hypergraphes partiels, présentés par la Figure \ref{hyp_part}, tel que $H_1$ ne contient que les hyperarêtes auxquelles appartient le sommet $7$ (dont le support est égal à $3$), $H_2$ ne contient que celles auxquelles appartient $4$ (dont le support est égal à $2$) et $H3$ contient les hyperarêtes restantes, i.e., celles qui renferment le sommet $1$. Il importe de noter qu'en choisissant $\{2, 4, 7\}$, au lieu de $\{1, 4, 7\}$, le résultat reste le même.

\begin{figure}
\begin{minipage}{.33\textwidth}\centering
\centering
\includegraphics[scale=0.7]{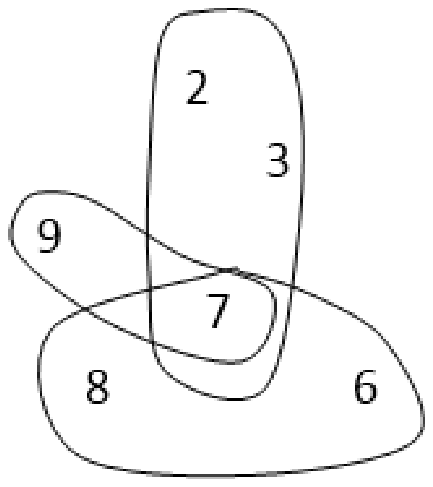}
\end{minipage}\hfill
\begin{minipage}{.3\textwidth}\centering
\centering
\includegraphics[scale=0.7]{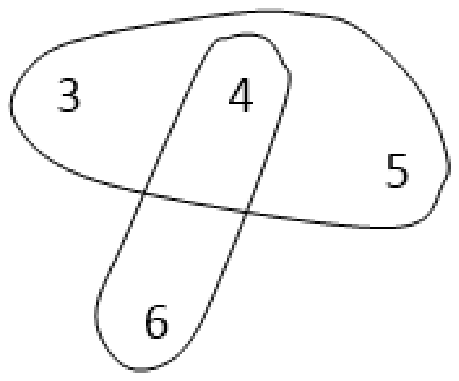}
\end{minipage}\hfill
\begin{minipage}{.36\textwidth}\centering
\centering
\includegraphics[scale=0.7]{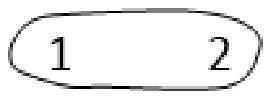}
\end{minipage}
\caption{Les 3 hypergraphes partiels dérivés de $H$ : $H_1$, $H_2$ et $H3$}\label{hyp_part}
 \end{figure}

\end{example}

\section{Traverses minimales locales : approche et algorithme}
Optimiser le calcul de ces traverses minimales revient donc principalement à réduire le nombre de candidats traités. Ceci passe par la réduction de la taille de l'hypergraphe d'entrée.
L'approche que nous proposons consiste à construire, à partir de l'hypergraphe d'entrée $H$, $k$ hypergraphes partiels ($H_1$, $H_2$, \ldots, $H_k$). Le calcul de l'ensemble des traverses minimales locales, $\mathcal{M}_{H_i}$ de chaque hypergraphe partiel $H_i$ s'en trouve amélioré puisque la taille de $H_i$ est relativement petite par rapport à celle de $H$. Ainsi, nous proposons d'effectuer l'union des hypergraphes partiels de façon à déterminer l'ensemble des traverses minimales $\mathcal{M}_{H}$ de $H$ à partir des k-uplets, minimaux au sens de l'inclusion, issus du produits cartésien des ensembles de traverses minimales locales déterminées pour les hypergraphes partiels $\mathcal{M}_{H_i}$. Dans ce qui suit, nous présentons l'algorithme \textsc{local-generator} dédié au calcul des traverses minimales et basé essentiellement sur les notions de nombre de transversalité et d'hypergraphe partiel.

\begin{algorithm}[!h]
\LinesNumbered
\Entree{Une matrice d'incidence $IM_H$ associée à $H$ = ($\mathcal{X}$, $\xi$)}
 \Sortie{$\mathcal{M}_{H}$, ensemble des traverses minimales de $H$}
 \Deb{
       $T$ = \textsc{GetMinTransversality2}($IM_H$);\\
       Ordonner les éléments de $T$ par ordre décroissant du support;\\
       $k$ = $\mid T \mid$;\\
       $i$ = $1$;\\
       \Tq{$i$ $\leq$ $k$}
       {
       $\xi_i$ = $\{e \in \xi \mid T[i] \in e\}$;\\
       $\mathcal{X}_i$ = $\mathcal{X} \cap \xi_i$;\\
       $\mathcal{M}_{H_i}$ = \textsc{mtminer}($H_i$);\\
       $i$ = $i$ + 1;\\
       }

       $\gamma_H$ = $\mathcal{M}_{H_1}$ $\times$ $\mathcal{M}_{H_2}$ $\times$ $\ldots$ $\times$ $\mathcal{M}_{H_k}$;\\
       \PourCh{$X$ $\in \gamma_H$}{
       \eSi{$\mid X \mid$ = $k$}{$\mathcal{M}_{H}$ = $\mathcal{M}_{H}$ $\cup$ \{$X$\};}{
       \Si{$\not\exists$ $x\in X:Supp(X)=Supp(X\backslash x)$}{$\mathcal{M}_{H}$ = $\mathcal{M}_{H}$ $\cup$ \{$X$\};\\}
       }
       }
    \Retour{($\mathcal{M}_{H}$)}
    }
     \caption{\textsc{local-generator}}\label{algo1}
\end{algorithm}
L'algorithme \textsc{local-generator}, dont le pseudo-code, est décrit par l'Algorithme \ref{algo1} prend en entrée une matrice d'incidence (correspondant à l'hypergraphe d'entrée) et fournit en sortie l'ensemble des traverses minimales. On suppose que les sommets de l'hypergraphe sont triés par ordre lexicographique. \textsc{local-generator} démarre par un appel à la fonction \textsc{GetMinTransversality}, dont le pseudo-code est décrit par l'Algorithme \ref{proc} du chapitre \ref{chap3} (page \pageref{proc}). Comme déjà mentionné, cette fonction calcule et retourne une traverse minimale dont la taille est minimale et le nombre de transversalité $k$ de l'hypergraphe. Ce dernier correspond désormais à la cardinalité de la traverse minimale retournée par la fonction. C'est à partir de cette traverse minimale retournée par \textsc{GetMinTransversality} que notre algorithme décompose l'hypergraphe d'entrée en des hypergraphes partiels.

Une fois la construction des $k$ hypergraphes partiels (lignes 7-8) effectuée, l'algorithme \textsc{local-generator} fait appel à un algorithme d'extraction des traverses minimales pour calculer leurs traverses minimales locales \footnote[7]{Dans les expérimentations, nous avons utilisé l'algorithme \textsc{MtMiner} pour accomplir cette tâche et nous remercions les auteurs de nous en avoir fourni une version.}, stockées dans $\mathcal{M}_{H_i}$ (ligne $9$). Etant donné que les hypergraphes $H_i$ sont relativement de petite taille, n'importe quel algorithme existant peut calculer les traverses minimales $\mathcal{M}_{H_i}$ en des temps très courts. Cet algorithme prend donc en entrée un hypergraphe partiel $H_i$ de $H$, dont l'ensemble des sommets $\mathcal{X}_i$ et l'ensemble des hyperarêtes $\xi_i$ ont été déjà calculés (lignes $7$ - $8$) et extrait, par niveaux, l'ensemble des traverses minimales locales à $H_i$ selon la définition \ref{Traverse2}.
A la fin de la boucle de la ligne $6$, \textsc{local-generator} a déjà préparé les ensembles des traverses minimales locales. Le produit cartésien (ligne $11$) de ces ensembles $\mathcal{M}_{H_i}$, permet de construire l'ensemble $\gamma_H$. Chaque élément de $\gamma_H$ issu de ce produit cartésien représente une traverse. Il reste à vérifier sa minimalité. Un des intérêts de notre décomposition de l'hypergraphe initial est d'éviter de tester la minimalité des éléments de $\gamma_H$ dont la cardinalité est égale à $k$. En effet, ces derniers représentent des traverses minimales de $H$ puisqu'il ne peut pas exister une traverse minimale de taille inférieure au nombre de transversalité de $H$. Pour les traverses de taille supérieure à $k$, \textsc{local-generator} teste la minimalité (lignes $15-16$) suivant la Proposition \ref{Traverse2}. Si le support d'un candidat $X$ est strictement supérieur au maximum des supports de ses sous-ensembles directs alors $X$ est une traverse minimale et est ajouté à $\mathcal{M}_{H}$.

\section{Etude de la complétude}
Le résultat du produit cartésien de deux ensembles de traverses minimales, $\mathcal{M}_H$ et $\mathcal{M}_G$, représente un ensemble de traverses. La minimalité est ensuite vérifiée par le biais de la condition d'essentialité de la Définition \ref{defessentiel}.
Dans le but de vérifier la complétude de notre approche, nous nous proposons de prouver par récurrence que $\mathcal{M}_{H}$ = Min \{$\mathcal{M}_{H_1}$ $\times$ $\mathcal{M}_{H_2}$ $\times$ $\ldots$ $\times$ $\mathcal{M}_{H_k}$\}, en s'inspirant de la Proposition \ref{KSprop}.
\begin{proposition}
Toute traverse minimale de l'hypergraphe $H$ peut être déduite à partir des ensembles de traverses minimaux des $H_i$, $i$ = $1 \ldots k$.
\end{proposition}

\begin{proof}
Soient $H_1 = (\mathcal{X}_1, \xi_1)$ et $H_2 = (\mathcal{X}_2, \xi_2)$ deux hypergraphes et $H$  = $H_1 \cup H_2$, avec $H = (\mathcal{X}, \xi)$ tel que $\mathcal{X} = \mathcal{X}_1 \cup \mathcal{X}_2$ et $\xi = \xi_1 \cup \xi_2$. D'après Berge \cite{berge1989}, l'ensemble des traverses minimales de $H$ est égal aux éléments, minimaux au sens de l'inclusion, appartenant au produit cartésien des ensembles minimaux de $H_1$ et $H_2$ :

$\mathcal{M}_{H}$ = $\mathcal{M}_{H_1 \cup H_2}$ = Min \{ $T$ $\mid$ $T \in$ $\mathcal{M}_{H_1}$ $\times$ $\mathcal{M}_{H_2}$ \} \textbf{(P1)}\\

Supposons que cette propriété est vraie pour l'union de $k$ hypergraphes :
$H'$ = $H_1 \cup H_2 \cup \ldots \cup H_k$ avec $H' = (\mathcal{X'}, \xi')$ et $H_i$ = ($\mathcal{X}_i$, $\xi_i$), $i$ = $1 \ldots k$.\\

Par hypothèse, nous avons alors :

$\mathcal{M}_{H'}$ = Min \{ $T$ $\mid$ $T \in$ $\mathcal{M}_{H_1}$ $\times$ $\mathcal{M}_{H_2} \times \ldots \times \mathcal{M}_{H_k}$ \}. \textbf{(P2)}\\
Montrons, à présent, que la propriété est vraie pour l'union de $k+1$ hypergraphes tel que $H$ = $H_1 \cup H_2 \cup \ldots \cup H_{k+1}$ où $\mathcal{X} = \mathcal{X}_1 \cup \mathcal{X}_2 \cup \ldots \cup \mathcal{X}_{k+1}$ et $\xi = \xi_1 \cup \xi_2 \cup \ldots \cup \xi_{k+1}$.\\
On a :

$H$ = $H_1 \cup H_2 \cup \ldots \cup H_{k+1}$ \\
$\Leftrightarrow$ $H$ = $H' \cup H_{k+1}$ où $H'$ = $\cup_{i=1}^k H_i$ et $H'$ = ($\mathcal{X'}$, $\xi'$) avec $\mathcal{X'} = \cup_{i=1}^k \mathcal{X}_i$ et $\xi' = \cup_{i=1}^k \xi_i$.\\

D'après \textbf{(P1)}, on a $\mathcal{M}_{H}$ = Min \{ $T$ $\mid$ $T \in$ $\mathcal{M}_{H'}$ $\times$ $\mathcal{M}_{H_{k+1}}$ \} et d'après \textbf{(P2)}, nous avons $\mathcal{M}_{H'}$ = Min \{ $T$ $\mid$ $T \in$ $\mathcal{M}_{H_1}$ $\times$ $\mathcal{M}_{H_2} \times \ldots \times \mathcal{M}_{H_k}$ \}. Nous obtenons, au final donc : $\mathcal{M}_{H}$ = Min \{ $T$ $\mid$ $T \in$ $\mathcal{M}_{H_1}$ $\times$ $\mathcal{M}_{H_2} \times \ldots \times \mathcal{M}_{H_{k+1}}$ \}.\\

Ceci prouve donc que toute traverse minimale de $H$, i.e., de l'ensemble $\mathcal{M}_{H}$, peut être calculée à partir des traverses minimales locales des $k$ hypergraphes partiels $H_i$, $1 \leq i \leq k$.
\end{proof}

\begin{proposition}
Toute traverse minimale $T$ déduite à partir des $\mathcal{M}_{H_i}$, tel que $T$ = min \{$T_1 \cup T_2 \cup T_3 \ldots \cup T_k$\} et $T_i$ $\in$ $\mathcal{M}_{H_i}$, est une traverse minimale de $H$.
\end{proposition}

\begin{proof}
$T_i \in \mathcal{M}_{H_i}$ $\Rightarrow$ $T_i \cap e_{ji} \neq \emptyset$ $\forall e_{ji} \in \xi_i$, $1\leq i\leq k$.\\
	
Étant donné que $\bigcup\limits_{i=1}^k \xi_{i}$ = $\xi$ et $\forall (\xi_i$, $\xi_j$) $\in$ $\xi \times \xi \backslash \{\xi_i\}$, $\xi_i \cap \xi_j$ = $\emptyset$, donc $\forall e \in \xi$, $\exists$ $T_i \subset T$ tel que $T_i \cap e \neq \emptyset$. Ainsi, $T$ est une traverse de $H$. (1).\\

Par vérification de la condition de minimalité de la Définition \ref{defessentiel} (ligne $16$ de l'algorithme \ref{algo1}), $T$ = min \{$T_1 \cup \ldots \cup T_k$\} donc $T$ est minimal dans $H$ (2).\\

(1) et (2) permettent de considérer $T$ comme une traverse minimale de $H$.

\end{proof}

\section{Etude Expérimentale}
Différentes expérimentations ont été réalisées sur des jeux de données variés afin d'évaluer l'algorithme \textsc{local-generator}.
 Le premier lot de jeux de données considérés correspond aux hypergraphes générés à partir des bases de données "\textit{Accidents}" \footnote[8]{\emph{http://archive.ics.uci.edu/ml}} et "\textit{Connect-4}" \footnote[9]{\emph{http://fimi.cs.helsinki.fi/data/}}, également utilisés dans le chapitre $4$. Le deuxième lot contient des hypergraphes aléatoires générés (à l'aide du générateur "\textit{random hypergraph generator}" implementé par Boros et \textit{al.} \cite{Boros03anefficient}), en fonction du nombre de sommets, du nombre d'hyperarêtes et de la taille minimale des hyperarêtes. De plus, au cours de notre étude expérimentale, nous avons pris soin de vérifir que la borne maximale retournée par la fonction \textsc{GetMinTransversality} est bien égale au nombre de transversalité pour chaque hypergraphe traité. Ceci implique que les traverses générées par un produit cartésien des différents ensembles de traverses minimales locales et composées d'un nombre de sommets égal à la valeur retournée par \textsc{GetMinTransversality} sont bien des traverses minimales.

\begin{table}[htbp]
\begin{center}
\centering
\begin{tabular}{|r|r|r||r|r||r|r|r|}
\hline
   &\emph{$|$$\mathcal{X}$$|$} & \emph{$|\xi|$}   & \emph{$\tau$(H)} & \emph{$|$$\mathcal{M}_H$$|$}  & \textsc{mmcs} & \textsc{ks} & \textsc{local-generator}\\
  \hline
    Accidents1  & 81   & 990    &  1 & 1 961 & 0,30 & 8,620 & 1,52 \\\hline
    Accidents2  & 336   & 10968    &  2   & 17 486 & 0,81 & - & 2,47\\\hline
    Connect-Win  & 79   & 12800  &  3 &  4 869 431 & 100,59 & - & 294,50 \\\hline
  \hline
  \end{tabular}
\end{center}
\caption{Caractéristiques et temps de traitement des hypergraphes Accidents et Connect (en secondes)}\label{carac10}
\end{table}

Les caractéristiques de chacun des hypergraphes du premier lot considéré sont rappellées dans le tableau \ref{carac10}. La première et la seconde colonne correspond, respectivement, au nombre de sommets et au nombre d'hyperarêtes des différents hypergraphes. La troisième colonne indique le nombre de transversalité, alors que la quatrième colonne indique le nombre de traverses minimales que renferme chaque hypergraphe.

Ces trois hypergraphes ont été traités par trois algorithmes : l'algorithme \textsc{mmcs} de \cite{uno1} l'algorithme \textsc{ks} de \cite{KavvadiasS05} et notre algorithme \textsc{local-generator}. Le tableau \ref{carac10} récapitule aussi les temps d'exécution de chaque algorithme sur chaque hypergraphe. L'algorithme \textsc{mmcs} étant déjà le plus rapide parmi tout ceux proposés dans la littérature, il l'est aussi sur ces trois jeux de données. \textsc{local-generator} est moins rapide alors que \textsc{ks} ne parvient pas à traiter les hypergraphes \emph{Accidents2} et \emph{Connect-Win}. Le fait que \textsc{local-generator} ait des temps de traitement plus grands que \textsc{mmcs} s'explique par le faible nombre de transversalité des 3 hypergraphes qui varie entre 1 et 3 comme indiqué dans le tableau \ref{carac10}. Dans ce cas, la décomposition de l'hypergraphe d'entrée en hypergraphes partiels, ne permet pas d'optimiser convenablement le calcul des traverses minimales. La stratégie "diviser pour régner" n'est pas pertinente lorsque la taille de la plus petite traverse minimale, d'un hypergraphe donné, est très petite. Extraire directement les traverses minimales sur l'hypergraphe considéré s'avère plus judicieux que de passer par les traverses minimales locales.

\begin{table}[htbp]
\begin{center}
\centering
\begin{tabular}{|r|r|r||r|r||r|r|r|}
\hline

    &\emph{$|$$\mathcal{X}$$|$} & \emph{$|\xi|$}  & \emph{$\tau$(H)} & \emph{$|$$\mathcal{M}_H$$|$} & \textsc{mmcs} & \textsc{ks} & \textsc{local-generator}\\
  \hline

    $H1$  & 96   & 52    &  8  & 832 564 740  & 2804,64 & 3911,431  & 1004,269\\ \hline
    $H2$  & 95   & 51  &  9   & 5 040 431 550 & 3608,182 & - & 2899,088\\ \hline
    $H3$  & 119   & 91    &   4 & 4 186 560 000 &  3115,226 & - & 1918,101\\ \hline
    $H4$  & 159   & 142   &   20  &  7 158 203 125 &  5509,455 & - &4775,364\\ \hline

  \hline
  \end{tabular}
\end{center}
\caption{Caractéristiques et temps de traitement des hypergraphes aléatoires (en secondes)}\label{carac30}
\end{table}

Le Tableau \ref{carac30} récapitule les caractéristiques des différents hypergraphes que nous avons générés. Ces données synthétiques ont été générées en fonction des probabilités minimale et maximale d'appartenance d'un sommet aux hyperarêtes dans l'hypergraphe. Si les nombres de sommets et d'hyperarêtes ne sont pas très élévés, ces hypergraphes renferment néanmoins un très grand nombre de traverses minimales qui varie entre $832$ $564$ $740$ et $7$ $158$ $203$ $125$. Le nombre de transversalité, $\tau$(H) variant de $4$ à $20$, est aussi élevé en comparaison avec les hypergraphes du Tableau \ref{carac10}. Ceci favorise donc notre approche puisque les hypergraphes d'entrée sont décomposés en un nombre important de petits hypergraphes partiels et, de ce fait, les traverses minimales de taille égale à $\tau$(H) y sont plus nombreuses épargnant ainsi à notre algorithme le test de la minimalité.

Les temps de traitement en secondes, récapitulés dans le tableau \ref{carac30}, montrent que l'algorithme \textsc{local-generator} présente des temps plus intéressants que ceux obtenus par les algorithmes \textsc{ks} et \textsc{mmcs}. Notons que l'algorithme de \cite{KavvadiasS05} ne parvient à extraire les traverses minimales que sur $H1$. Les temps d'exécution supérieurs à \emph{1500} secondes peuvent s'expliquer par le nombre élevé de traverses minimales calculées. L'écart entre \textsc{local-generator} et \textsc{mmcs} varie entre \emph{709} et \emph{1800} secondes. La différence de performances entre les tableaux \ref{carac10} et \ref{carac30} permet de dresser un profil des types d'hypergraphes sur lesquels notre approche est plus performante. En effet, \textsc{local-generator} présente des temps intéressants dès que la taille des plus petites traverses minimales de l'hypergraphe d'entrée est élevée ce qui lui permet de décomposer l'hypergraphe en plusieurs hypergraphes partiels.

De plus, le nombre de traverses minimales doit être important. Sur des hypergraphes renfermant peu de traverses minimales, \textsc{local-generator} peine à se montrer efficace puisque le nombre de traverses minimales devient négligeable par rapport au nombre de candidats traités et testés. De plus, le gain en temps de traitement est aussi conditionné par le nombre des plus petites traverses minimales. En effet, pour ces dernières, notre algorithme n'effectue pas de test de la minimalité et permet donc d'optimiser les temps de traitements nécessaires pour le calcul de toutes les traverses minimales.
\section{Conclusion}
Dans ce chapitre, nous avons introduit une nouvelle approche pour le calcul des traverses minimales d'un hypergraphe. Cette approche repose sur le paradigme "\emph{diviser pour régner}" afin de décomposer l'hypergraphe d'entrée en hypergraphes partiels, en fonction du nombre de transversalité. Le calcul des traverses minimales locales, correspondantes à ces hypergraphes partiels, permet de retrouver l'ensemble des traverses minimales à travers un produit cartésien combiné à un test de la minimalité. Ceci nous a permis d'introduire un nouvel algorithme \textsc{local-generator} pour l'extraction des traverses minimales. L'étude expérimentale a confirmé l'intérêt de notre approche sur un type précis d'hypergraphes renfermant des propriétés données. Cette approche a été publiée dans une conférence avec comité de lecture \cite{NidhalEGC2014}.

\chapter*{Conclusion générale }
\markboth{Conclusion générale}{\textsc{Conclusion générale}}
\addcontentsline{toc}{chapter}{Conclusion générale} \indent Les
notations semi formelles permettent de spécifier le système en
fournissant ainsi un bon support de communication avec
l'utilisateur alors que les notations formelles apportent une
précision à la spécification. Une précision primordiale pour tout
raisonnement de vérification. L'idée d'intégrer les méthodes semi
formelles et les méthodes formelles permet de bénéficier de la
sémantique des premières et de rendre plus accessibles l'accès des
utilisateurs aux deuxièmes. \\
\indent Parmi les travaux, qui ont vu le jour ces dernières
années, intégrant les méthodes formelles aux méthodes semi
formelles, nous nous sommes intéressés à ceux mettant en jeu la
méthode \textsc{b} et le langage \textsc{uml}. La traduction des
modèles \textsc{uml} en des spécifications \textsc{b} engendre des
machines \textsc{b} compliquées et difficiles à comprendre. Dans
ce contexte, les méthodes semi-formelles peuvent apporter des
solutions à des problèmes qui n'ont pas été résolus par les
méthodes formelles. En effet, une des faiblesses de la méthode
\textsc{b} consiste en son manque de moyens de structuration. Ceci
rend bien évidemment difficile la compréhension des modèles
\textsc{b} surtout pour un utilisateur non
habitué à ses notations.\\
\indent Il paraît donc utile de prévoir des méthodes de conception
où les méthodes formelles et semi-formelles coexistent en offrant
deux vues complémentaires d'un modèle commun sous-jacent et en
permettant au concepteur de travailler sur la vue la plus adaptée.
L'approche proposée dans \cite{FKIH} est une démarche dans ce sens
dans le mesure où elle permet de générer des diagrammes de classes
à partir des machines \textsc{b}. Cette approche \cite{FKIH}
consiste en un nombre de règles de transformation. Dans le but
d'automatiser la traduction d'un modèle \textsc{b} en un diagramme
de classes \textsc{uml}, nous avons mis en oeuvre un nouveau
moteur d'inférence, nommé \textsc{Thinker}, qui à partir d'un
modèle \textsc{b} initial permet de générer le diagramme de
classes correspondant. Cependant, cet objectif d'automatisation
est conditionné par quelques contraintes
auquel \textsc{Thinker} a dû faire face.\\
\indent La principale originalité de notre outil constitue donc à
donner la main à l'utilisateur chaque fois que ceci est
nécessaire. En effet, certaines règles de transformation de
l'approche \cite{FKIH} nécessite l'intervention de l'utilisateur
pour choisir un chemin de traduction parmi plusieurs. Au cours du
processus de transformation, il existe des règles qui offrent
plusieurs representations possibles et une seule peut être
appliquée. Le choix d'une representation est tributaire du seul
utilisateur et \textsc{Thinker} intègre cette option dans son
mécanisme d'inférence. A la rencontre d'une règle ambiguë qui
offre plusieurs représentations possibles, notre moteur
d'inférence donne immédiatement la main à l'utilisateur. Ce
dernier, en s'aidant des représentations graphiques en
\textsc{uml} des différents traductions possibles de la règle en
question, choisit celle qui correspond le mieux à ses besoins.
Autre particularité de \textsc{Thinker}, dans le cas où
l'utilisateur  exprimerait le souhait de revenir sur un choix
antérieur pour le changer, le moteur d'inférence offre la
possibilité de rebrousser chemin et revenir à un choix deja pris
pour en essayer un autre, et ce à n'importe quel moment du
processus d'inférence. Ceci permet à l'utilisateur de diriger lui
même l'application des règles en contrôlant le résultat final de
la base des faits. Les résultats obtenus sont cohérents et sont
conformes à la volonté de l'utilisateur, d'autant plus, dans le
cas d'une transformation de \textsc{b} vers \textsc{uml}.\\
\indent D'autre part, \textsc{Thinker} intègre un parseur qui
permet d'extraire la base des faits initiale directement à partir
d'un modèle \textsc{b} donné et ce sans passer par l'utilisateur.
Ainsi, n'importe quel personne peut utiliser notre outil sans
forcément beaucoup y connaître en méthode \textsc{b}. En outre,
\textsc{Thinker} peut se targuer d'être générique. Outre le cas de
la transformation des méthodes \textsc{b} en des diagrammes
\textsc{uml}, notre moteur d'inférence est capable de déclencher
un mécanisme d'inférence à partir de n'importe quelles base des
faits et bases des règles. Nous avons ainsi crée un nouveau moteur
d'inférence qui se distingue de ses prédécesseurs par un haut
degré d'interactivité avec l'utilisateur. Cette interactivité
garantit un résultat final cohérent parfaitement dirigé par
l'utilisateur tout au long du processus d'inférence. De plus, le
fait que les différentes représentations possibles en \textsc{uml}
soient dessinées par notre moteur d'inférence donne une idée
encore plus précise, à l'utilisateur, des conclusions offertes par
chaque règle
constituant une ambiguïté donnée.\\
\indent Grâce à \textsc{Thinker}, notre objectif d'automatisation
de l'approche de transformation des méthodes \textsc{b} en des
diagrammes des classes \textsc{uml} a été atteint et une étude de
cas est présentée dans le dernier chapitre.\\
\indent Notre travail a dégagé quelques perspectives. En effet,
nous n'avons pris en compte dans notre travail que la traduction
d'un modèle \textsc{b} en un diagramme de classes \textsc{uml}.
Nous pouvons étendre les fonctions de \textsc{Thinker} pour qu'il
puisse générer en plus un diagramme de séquences à partir d'une
spécification \textsc{b} donnée. Par la suite, nous  pouvons
envisager de mettre en place une aide à la décision pour le choix
d'une solution en cas d'ambiguïté. Cette extension permettrait à
l'utilisateur de se faire conseiller un choix qui ressemble le
plus à ses choix antérieurs. Cette idée rentre dans l'optique
d'une volonté d'automatiser encore plus le processus de
transformation de \textsc{b} vers \textsc{uml}.

\nocite{ref1}
\nocite{ref2}
\nocite{ref3}
\nocite{ref4}
\nocite{ref5}
\nocite{ref6}
\nocite{ref7}
\nocite{ref8}
\nocite{ref9}
\nocite{ref10}
\nocite{ref11}
\nocite{ref12}
\nocite{ref13}
\nocite{ref14}
\nocite{ref15}
\nocite{ref16}
\nocite{ref17}
\nocite{ref18}
\nocite{ref19}
\nocite{ref20}

\bibliographystyle{alpha}
\bibliography{2014Jelassi}

\newpage
\include{abstract}
\end{document}